\documentclass[reqno]{amsart}

\usepackage{amsmath,amssymb,amsfonts, wasysym}
\usepackage[utf8]{inputenc}

\usepackage{amsthm, amsmath}
\usepackage{url}

\usepackage{multicol}

\usepackage{graphicx}
\usepackage{caption}
\usepackage{subcaption}

\usepackage{thmtools,thm-restate}
\usepackage{upgreek}

\usepackage{booktabs} 

\usepackage[ruled]{algorithm2e} 
\usepackage{hyperref}
\usepackage{color}
\usepackage{enumitem}
\numberwithin{equation}{subsection}

\newtheorem{theorem}{Theorem}[section]
\newtheorem{lemma}[theorem]{Lemma}
\newtheorem{definition}[theorem]{Definition}
\newtheorem{proposition}[theorem]{Proposition}
\newtheorem{corollary}[theorem]{Corollary}

\allowdisplaybreaks

\begin{document}
\title{A First-Order Logic for Reasoning about Knowledge and Probability}

\author{Siniša Tomović}
\address{Faculty of Technical Sciences, Novi Sad, Serbia}
\email{\url{sinisatom@turing.mi.sanu.ac.rs}}

\author{Zoran Ognjanović}
\address{Mathematical Institute of the Serbian Academy of Sciences and Arts, Belgrade, Serbia}
\email{\url{zorano@mi.sanu.ac.rs}}

\author{Dragan Doder}
\address{Universit\' e Paul Sabatier -- CNRS, IRIT, France}
\email{\url{dragan.doder@irit.fr}}

\maketitle
\begin{abstract}
We present a first-order probabilistic epistemic logic, which allows combining 
operators of knowledge  and probability within a group of possibly infinitely many agents. The proposed framework is the first order extension of the logic of Fagin and Halpern from   (J.ACM 41:340-367,1994).
We define its syntax and semantics, and prove the strong completeness property of the corresponding axiomatic system.\footnotemark

\end{abstract}

\footnotetext{This paper is revised and extended version of the conference paper \cite{Sini} presented at the Thirteenth European Conference on Symbolic and Quantitative Approaches to Reasoning with Uncertainty (ECSQARU 2015), in which we introduced the propositional variant of the logic presented here, using a similar axiomatization technique.}
\section{Introduction}

Reasoning about knowledge is  widely used in many applied fields such as computer science, artificial intelligence, economics, game theory etc \cite{aumann1976,Knjiga,Gean,Tuttle}. A particular line of research concerns the formalization 
in terms of multi-agent epistemic logics, that speak about  knowledge about facts, but also about knowledge of other agents. 
One of the central notions is that of common knowledge, which has been shown as crucial for a variety of applications dealing with reaching agreements or coordinated actions \cite{Distr}.
Intuitively,  $\varphi$ is common knowledge of a group of agents exactly when everyone knows
that everyone knows that everyone knows\dots that $\varphi$
is true.

%
%
%
%
%
%
%

However, it has been shown that
in many practical systems common knowledge cannot be attained \cite{Distr,DBLP:journals/apal/FaginHMV99}.
This motivated some researchers to consider a weaker variant that still may
be sufficient for carrying out a number of coordinated actions \cite{DBLP:journals/dc/PanangadenT92,Tuttle,DBLP:journals/jolli/MorrisS97}. 
One of the approaches proposes a probabilistic variant of common knowledge \cite{Monderer}, which assumes that coordinated actions hold with high probability.
A propositional logical system which formalizes that notion is presented in \cite{KP}, where Fagin and Halpern developed a joint framework for reasoning about knowledge and probability and proposed a complete axiomatization.

\smallskip

We use the paper  \cite{KP} as a starting point and generalize it in two ways:

\smallskip

First,  we extend  the propositional formalization from \cite{KP} by allowing reasoning about knowledge and probability of events expressible in \emph{a first-order language}. We use the most general approach, allowing arbitrary combination of standard epistemic operators, probability operators, first-order quantifiers and, in addition, of probabilistic common knowledge operator. 
The need for first-order extension is recognized by epistemic and probability logic communities. Wolter \cite{Wolter} pointed out that first-order common knowledge logics are of interest both from the point of
view of applications and of pure logic. He argued that first-order base is necessary whenever
the application domains are infinite (like in epistemic analysis
of the playability of games with mixed strategies), or finite, but with the cardinality of which is not known in advance, which is a frequent case in in the field of Knowledge Representation.
Bacchus \cite{DBLP:journals/ci/Bacchus90a} gave the similar argument in the context of probability logics,
arguing that, while a domain may be finite, it is questionable if there is a
fixed upper bound on its size,  and he also pointed out that there are many domains, interesting
for AI applications, that are not finite.


Second, we consider  \emph{infinite number of agents}. 
While this assumption is not of interest in probability logic, it was studied in epistemic logic. 
Halpern and Shore \cite{InfiHalp} pointed out that economies, when regarded as  teams in a game, are often modeled as having infinitely many  agents and that such modeling in epistemic logic is also convenient in the situations where the group of agents and its upper limit are not known apriori. 

\smallskip


The semantics for our logic consists of Kripke models enriched with probability spaces. Each possible world contains a first order structure, each agent in each world is equipped with a set of accessible worlds and a finitely additive probability on measurable sets of worlds.
In this paper we consider the most general semantics, with independent modalities for knowledge and probability. Nevertheless, in Section \ref{sec canonical} we show  how to modify the definitions and results of our logic, in order to capture some interesting relationships between the  modalities for knowledge and probability (previously considered in \cite{KP}), especially the semantics in which agents assign probabilities only to the sets of worlds they consider possible.

\smallskip

The main result of this paper is a sound and \emph{strongly complete} (``every consistent set of sentences is satisfiable") axiomatization. The negative result of Wolter \cite{Wolter}
shows that there is no finite way to
axiomatize first order common knowledge logics, and 
that infinitary axiomatizations  are the best we can do (see Section \ref{sec ax issues}).
We obtain completeness using  infinitary rules of inference. Thus, formulas are finite, while only proofs are allowed to be (countably) infinite. 
We use a Henkin-style  construction of saturated extensions of consistent theories. From the technical point of view, we  modify some of our earlier developed methods presented in \cite{Doder,Milos,Ognj4,PLogics}.\footnote{For the detailed overview of the approach, we refer the reader to \cite{DBLP:books/sp/OgnjanovicRM16}. A similar approach is later used in \cite{DBLP:journals/logcom/Zhou09}.} Although we use an alternative axiomatization for the epistemic part of logic (i.e., different from original axiomatization given in \cite{KP,Guide}),  we prove that standard axioms are derivable in our system.

\smallskip

There are several papers on completeness of epistemic logics with common knowledge. 

In \emph{propositional case}, a  finitary axiomatization, which is \emph{weakly complete} (``every consistent formula is satisfiable''), is obtained by Halpern and Moses \cite{Guide} using a fixed-point axiom for common knowledge. 
On the other hand, strong completeness for any finitary axiomatization is impossible, due to lack of   compactness  (see Section \ref{sec ax issues}).  
Strongly complete axiomatic systems are proposed in \cite{DBLP:journals/rml/Tanaka03,Kooi}. They contain an infinitary inference rule, similar to one of our rules\footnote{It is easy to check that our inference rule RC from Section \ref{section ax} generalize the rule from \cite{DBLP:journals/rml/Tanaka03,Kooi}, due to presence of probability operators.}, for capturing semantic relationship between the operators of group  knowledge and common knowledge.

In \emph{first-order case}, the set of valid formulas is not recursively enumerable \cite{Wolter} and, consequently, there is no complete finitary axiomatization. 
One way to overcome this problem is by including infinite formulas in the language as in \cite{Tanaka}. 
A logic with finite formulas, but an infinitary inference rule, is proposed in \cite{Kaneko2002}, while a Genzen-style axiomatization with an inifinitary rule is presented in  \cite{DBLP:journals/rml/Tanaka03}.
On the other hand, a finitary axiomatization of monadic fragments of the logic, without function symbols and  equality, is proposed in \cite{Sturm2002}.

\smallskip

Fagin and Halpern \cite{KP} proposed a joint frame for reasoning about knowledge and probability. 
Following the approach from \cite{HalpProb}, they extended the propositional epistemic language with formulas which express linear combinations of probabilities, i.e., the formulas of the form  $a_1p(\varphi_1)+ ... + a_kp(\varphi_k) \geq b$, where $a_1,.., a_k, b \in \mathbb Q$, $k \geq 1$. They proposed a finitary axiomatization and proved weak completeness, using the small model theorem. 
%
%
Our axiomatization technique is different. Since in the first order case a complete finitary axiomatization is not  possible, we use infinitary rules and we prove strong completeness using Henkin-style method. We use unary probability operators and we axiomatize the probabilistic part of our logic following the techniques from \cite{DBLP:books/sp/OgnjanovicRM16}. In particular, our logic incorporates the single-agent probability logic $LFOP_1$ from \cite{PLogics}. However, our approach can be easily extended to include linear combinations of probabilities, similarly as it was done in \cite{DBLP:conf/uai/DoderO15,DBLP:journals/amai/OgnjanovicMRDP12}.

\smallskip

We point out that all the above mentioned  logics do not support infinite group of agents, so the group knowledge operator is defined as the conjunction of knowledge of individual agents.   A weakly  complete axiomatization  for common knowledge with infinite number of agents (in  non-probabilistic setting) is presented in \cite{InfiHalp}. In our approach, the knowledge operators of groups and individual agents are related via an infinitary rule (RE from Section \ref{section ax}).

\smallskip

The rest of the paper is organized as follows: In Section 2 we introduce Syntax and Semantics. Section 3 provides the axiomatization of our logic system, followed by the proofs of its soundness. In Section 4 we prove several theorems, including Deduction theorem and Strong necessitation. The completeness result is proven in Section 5. Section 6 we consider an extension of our logic by incorporating \emph{consistency condition} \cite{KP}.
The concluding remarks are given in  Section 7.

\section{Syntax and sematics}

In this section we present the syntax and semantics of our logic, that we call $\mathcal PCK^{fo} $.\footnote{$\mathcal PCK  $ stands for ``probabilistic common knowledge'', while $fo $ indicates that our logic is a first-order logic.}
Since the main goal of this paper is to combine the epistemic first order logic with reasoning about probability, our language extends a  first order language with both epistemic  operators, and the operators for reasoning about probability and probabilistic knowledge.
We introduce the set of formulas based on this language and the corresponding possible world  semantics, and we define the  satisfiability relation.

\subsection{Syntax}

Let $[0,1]_{\mathbb Q}$ be the set of rational numbers from the real interval $[0,1]$, $\mathbb N$ the set of non-negative integers, $\mathcal A$ an at most countable set of agents, and $\mathcal G$ a countable set of nonempty subsets of $\mathcal A$.

The language $\mathcal L_{PCK^{fo}} $ of the logic $\mathcal PCK^{fo} $ contains: 
\begin{itemize}
	\item 
	a countable set of  variables $Var = \{ x_1, x_2, \dots\}$, 
	
	\item $m$-ary relation symbols $R_0^m, R_1^m, \dots$ and function symbols  $f_0^m, f_1^m, \dots$ for every integer $m \geq 0$, 
	
	\item Boolean connectives $\wedge$ and $\neg$, and the first-order quantifier $\forall$, 
	
	\item unary  modal knowledge operators $K_i, E_G, C_G$, for every  $i \in \mathcal A$ and  $G \in \mathcal G$,
	
	\item unary   probability operator $P_{i, \geq r}$ and the operators for  probabilistic knowledge $E_G^r $ and  $C_G^r $, where $i \in \mathcal A$, $G \in \mathcal G$, $r \in [0,1]_{\mathbb Q}$. 
	
\end{itemize}

By the standard convention, constants are $0-$ary function symbols.
Terms  and atomic formulas are defined in the same way  as in the classical first-order logic. 

\begin{definition}[Formula]
	The set of formulas $For_{PCK^{fo}} $ is the least set  containing all atomic formulas such that: if  $\varphi, \psi \in \mathcal For_{PCK^{fo}} $ then $\neg \varphi$, $\varphi \wedge \psi$, $K_i \varphi$, $E_G \varphi$, $C_G \varphi$, $E_G^r \varphi$, $C_G^r \varphi$, $P_{i, \geq r}\varphi \in \mathcal For_{PCK^{fo}}  $, for every $i \in \mathcal A$, $G \in \mathcal G$ and  $r \in [0,1]_{\mathbb Q}$.
	%
\end{definition}

We use the standard abbreviations to introduce  other Boolean connectives  $ \ \to$, $\vee$ and $\leftrightarrow$, the quantifier $\exists$ and the symbols $\perp, \top$.
We also introduce the operator $K_i^r$ (for  $i \in \mathcal A$  and  $r \in [0,1]_{\mathbb Q}$) in the following way: the formula $K_i^r \varphi$ abbreviates  $K_i(P_{i,\geq r}{\varphi})$.

The meanings of the operators of our logic are  as follows.
\begin{itemize}
	\item $K_i\varphi$ is read as \emph{``agent i knows $\varphi$''} and   $E_G\varphi$ as \emph{``everyone in the group $G$ knows $\varphi$''}. The formula $C_G\varphi$ is read \emph{``$\varphi$ is common knowledge among the agents in $G$''},  which means that  everyone (from $G$) knows $\varphi$, everyone knows that
	everyone knows $\varphi$, etc.
	
	
	\smallskip
	
	\noindent\emph{Example.}
	The sentence \emph{``everyone in the group $G$ knows that if agent $i$ doesn't know $\varphi$, then $\psi$ is common knowledge in  $G$''}, is written as  $$E_G ( \neg K_i \varphi \to C_G \psi ).$$

	\item The probabilistic formula $P_{i,\geq r}{\varphi}$ says that \emph{the probability that formula $\varphi$ holds is at least $r$ according to the agent $i$}.

	\smallskip

	\item	$K_i^r \varphi$ abbreviates the formula $K_i(P_{i,\geq r}{\varphi})$. It means that \emph{agent $i$ knows that the probability of $\varphi$ is at least $r$}.
	
	\smallskip
	
	\noindent\emph{Example.} Suppose that agent $i$ considers two only possible scenarios for an event $\varphi$, and that each of these scenarios puts a different probability space on events. In the first scenario, the probability of $\varphi$ is $1/2$, and in the second one it is  $1/4$. Therefore, the agent knows that probability of $\varphi$ is at least $1/4$, i.e., $K_i(P_{i,\geq 1/4}{\varphi})$.
	
	\smallskip
	
	\item 	$E_G^r \varphi$ denotes that \emph{everyone in the group $G$ knows that the probability of $\varphi$ is at least $r$}. Once $K_i^r \varphi$ is introduced,  $E_G^r$ is defined as  a straightforward probabilistic  generalization of the operator $E_G$.
	
	\smallskip
	
	\item 	 $C_G^r \varphi$ denotes that it is a \emph{common knowledge in the group $G$ that the probability of $\varphi$ is at least $r$}. For a given threshold $r\in [0,1]_{\mathbb Q}$, $C_G^r $ represents a generalization of non-probabilistic operator $C_G $.
	
	\smallskip
	
	\noindent\emph{Example.} The formula $$E_G^s ( K_i (\exists x) \varphi(x) \wedge \neg C_G^r \psi )$$ says that  \emph{everyone in the group $G$ knows that the probability that both   agent $i$ knows that $\varphi(x)$ holds for some $x$, and that  $\psi$ is not common knowledge among the agents in $G$ with probability at least $r$, is at least $s$}.

\end{itemize}

Note that the other types of probabilistic operators can also be introduced as abbreviations:  $P_{i,<r}{\varphi}$ is $ \neg P_{i,\geq r}{\varphi}$,  $P_{i,\leq r}{\varphi}$ is $ P_{i,\geq 1-r}{\neg \varphi}$,  $P_{i,>r}{\varphi}$ is $\neg P_{i,\leq r}{\varphi}$ and $P_{i,=r}{\varphi}$ is $P_{i,\leq r}{\varphi} \wedge P_{i,\geq r}{\varphi} $.

\smallskip

Now we define what we mean by a \emph{sentence} and a \emph{theory}. The following definition uses the notion \emph{free variable}, which is defined in the same way  as in the classical first-order logic. 
\begin{definition}[Sentence]
	A formula with no free variables is called a sentence. The set of all sentences is denoted by $Sent_{PCK^{fo}}$. A set of  sentences  is called  \emph{theory}.
\end{definition}

Next we introduce a special kind of formulas in the implicative form, called \emph{$k$-nested implications}, which will have an important role in our axiomatization. 

\begin{definition}[$k$-nested implication]\label{def k nested}
	Let   $\tau \in For_{PCK^{fo}}$ be a formula and let and $k\in \mathbb N$.
	Let 	 $\boldsymbol{\uptheta} =(\theta_{0}, \dots, \theta_{k})$  be a sequence of $k$ formulas, and $\mathbf{X} = (X_{1}, \dots, X_{k})$ a sequence of  knowledge and probability operators  from $ \{ \mathcal K_i \, | \, i \in \mathcal A\} \cup \{ P_{i,\geq 1} \, | \, i \in \mathcal A\}$.
	The \emph{$k$-nested implication} formula $\Phi_{k, \boldsymbol{\uptheta}, \mathbf{X}}(\tau)$  is  defined inductively, as follows:

	\[
	\Phi_{k, \boldsymbol{\uptheta}, \mathbf{X}}(\tau) =\left\{
	\begin{array}{ll}
	\theta_0 \to \tau, \,\, k=0\\
	\theta_{k} \to X_{k}\Phi_{k-1}(\tau,(\theta, X)_{j=0}^{k-1}), \,\, k \geq 1.
	\end{array}
	\right.
	\]
	
\end{definition}
%
%
%



For example, if $\mathbf{X} = (K_a, P_{b,\geq 1}, K_c)$,  $a,b,c \in \mathcal A$, then

\[
\Phi_{3, \boldsymbol{\uptheta}, \mathbf{X}}(\tau) = \theta_3 \to K_c( \theta_2 \to P_{b,\geq 1}(\theta_1 \to K_a(\theta_0 \to \tau))).
\]

The structure of these $k$-nested implications is shown to be  convenient for the proof of Deduction theorem (Theorem \ref{deduc}) and Strong necessitation theorem (Theorem \ref{thm strong necessitation}).


\subsection{Semantics}
\label{sec semantics}

The semantic approach for $PCK^{fo}$  extends the classical possible-worlds model for epistemic logics, with probabilistic spaces. 

\begin{definition}[$PCK^{fo}$ model]\label{def model}
	
	A $PCK^{fo}$ model is  a \emph{Kripke structure for knowledge and probability} 
	which is represented by a tuple $$M=(S,D,I,\mathcal  K, \mathcal  P),$$ where:
	\begin{itemize}
		\item $S$ is a nonempty set of \emph{states} (or \emph{possible worlds})
		
		\item $D$ is a nonempty domain
		
		\item $I$ associates an interpretation $I(s)$ with each state $s$ in $S$ such that for all $i \in \mathcal A$ and all $k, m \in \mathbb N$:
		
		\begin{itemize}
			
			\item  $I(s)(f^m_k)$ is a function from $D^m$ to $D$,
			
			\item for each $s' \in S$, $I(s')(f^m_k) = I(s)(f^m_k)$
			
			\item $I(s)(R^m_k)$ is a subset of $D^m$,
			
		\end{itemize}
		
		\item $\mathcal K = \{\mathcal K_i \, | \, i \in \mathcal A \}$ is a set  of binary relations on $S$. We denote $\mathcal K_i(s){\buildrel def\over=} \{  t \in s \, | \, (s,t) \in \mathcal K_i     \}$, and write $s \mathcal K_i t$ if $t \in \mathcal K_i(s)$.
		
		\item $\mathcal P$ associates to every agent $i \in \mathcal A$ and every state $s \in S$ a probability space $\mathcal P (i,s) = (S_{i, s}, \chi_{i, s}, \mu_{i, s})$, such  that

		\begin{itemize}

			\item $S_{i, s}$ is a non-empty subset of $S$,

			\item $\chi_{i, s}$ is an algebra of subsets of $S_{i, s}$, whose elements are called \emph{measurable sets}, and

			\item $\mu_{i, s}: \chi_{i, s}\to [0,1]$ is a finitely-additive probability measure ie. 
			
			\begin{itemize}
				
				
				\item $\mu_{i, s}(S_{i, s})=1$ and 
				
				\item $\mu_{i, s}(A \cup B) = \mu_{i, s}(A) +\mu_{i, s}(B)$ if $A \cap B = \emptyset, A,B \in \chi_{i, s}$.
				
			\end{itemize}
		\end{itemize}
	\end{itemize}
	
\end{definition}

In the previous definition we assume that  the domain is fixed (i.e., the domain is  same in all the worlds) and that 
the terms are rigid, i.e., for every model their meanings are the same in all worlds. Intuitively, the first assumption means that it is common knowledge which objects
exist. Note that the second assumption implies that it is common knowledge which
object a constant designates. As it is pointed out in \cite{Sturm2002}, the first  assumption is natural for all those application domains that
deal not with knowledge about the existence of certain objects, but rather with knowledge about facts.
Also, the two  assumptions  allow us to give
semantics of probabilistic formulas which is similar to the objectual interpretation
for first order modal logics
\cite{Garson2001}.

Note that those standard assumptions for modal logics are essential to ensure  validity of all first-order axioms.
For example, if the terms are not rigid, the classical first order axiom 
$$\forall \varphi(x) \to \varphi(t),$$ where  the term $t$ is free for $x$ in $\varphi$,
would not be valid (an example is given in 
\cite{DBLP:journals/ai/Halpern90}). Similarly, Barcan formula (axiom FO3 in Section \ref{section ax}) holds only for fixed domain models.

\smallskip


For a model $M = (S,D,I,\mathcal  K, \mathcal  P)$ be a $PCK^{fo}$, the notion of \textit{variable valuation} is defined in the usual way: a variable valuation $v$ is a function which assigns the elements of the domain to the variables, ie., $v: Var \to D$. If $v$ is a valuation, then $v[d/x]$ is a valuation identical to $v$, with exception that $v[d/x](x)=d$.
\begin{definition}[Value of a term]
	The value of a term $t$ in a state $s$ with respect to $v$, denoted by $I(s)(t)_v$, is defined in the following way:
	\begin{itemize}
		
		\item if $t \in Var$, then $I(s)(t)_v=v(t)$,
		\item if $t = F^{k}_j(t_1, \dots, t_k)$, then $I(s)(t)_v = I(s)(F^{k}_j)(I(s)(t_1)_v, \dots, I(s)(t_k)_v)$.
		
	\end{itemize}
\end{definition}

The next definition will use the following knowledge operators, which we introduce in the inductive way:

\begin{itemize}
	\item $(E_G)^1 \varphi = E_G \varphi$
	\item $(E_G)^{m+1} \varphi=E_G((E_G )^k \varphi)$, $m \in \mathbb N$
	\item $ (F_G^r)^0 \varphi = \top$
	\item $ (F_G^r)^{m+1} \varphi = E_G^r(\varphi \wedge (F_G^r)^{m} \varphi)$, $m \in \mathbb N$.
\end{itemize}

Now we define \emph{satisfiability} of formulas from in the states of introduced models.

\begin{definition}[Satisfiability relation]\label{satdef}
	\emph{Satisfiability} of formula $\varphi$ in a state $s \in S$ of a model $M$, under a valuation $v$, denoted by $$(M,s,v) \models \varphi,$$ is defined in the following way:
	
	\begin{itemize}
		
		\item $(M,s,v) \models P^{k}_j(t_1, \dots, t_k)$ iff $(I(s)(t_1)_v, \dots, I(s)(t_k)_v) \in I(s)(P^{k}_j)$
		
		\item $(M,s,v)\models \neg \varphi$   iff    $(M,s,v)\not\models \varphi$
		
		\item $(M,s,v)\models \varphi \wedge \psi$   iff   $(M,s,v)\models \varphi$  and   $(M,s,v)\models \psi$
		
		\item $(M,s,v)\models (\forall x)\varphi$ iff for every $d \in D$, $  (M,s,v[d/x])\models \varphi$
		
		\item  $(M,s,v)\models K_i \varphi$   iff   $(M,t,v)\models \varphi$     for all $t \in \mathcal  K_i(s)$
		
		\item  $(M,s,v)\models E_G \varphi$   iff  $(M,s,v)\models K_i \varphi$  for all  $i \in G$
		
		\item   $(M,s,v)\models C_G \varphi$   iff   $(M,s,v)\models (E_G)^m \varphi$  for every $m \in \mathbb N$

		\item $(M,s,v)\models P_{i,\geq r}{\varphi}$  iff $\mu_{i, s} (\{ t \in S_{i, s} \, | \, (M,t,v) \models \varphi\})  \geq r$

		\item $(M,s,v)\models E_G^r \varphi$ iff $(M,s,v) \models K_i^r \varphi$ for all $i \in G$

		\item $(M,s,v)\models C_G^r \varphi$   iff  $(M,s,v)\models (F_G^r)^m \varphi$   for every $m \in \mathbb N$
		
	\end{itemize}
	
\end{definition}

\paragraph{Remark.}
The semantic definition of the probabilistic common knowledge operator $C_G^r$ from the last item of Definition \ref{satdef} is first proposed by Fagin and Halpern in \cite{KP}, as a generalization of the operator $C_G$ regarded as the infinite conjunction of all degrees of group knowledge. It is important to mention that this is not the only proposal for generalizing the nonprobabilistic case. Monderer and Samet \cite{Monderer} proposed a more intuitive definition, where probabilistic common knowledge is semantically  equivalent to the infinite conjunction of the formulas $E_G^r \varphi, (E_G^r)^2 \varphi, (E_G^r)^3 \varphi \dots $
Although both are legitimate probabilistic generalizations, in this paper we accept the definition of Fagin and Halpern \cite{KP}, who argued that their proposal  seems more adequate for the analysis of problems like probabilistic coordinated attack and Byzantine agreement protocols \cite{Tuttle}.
As we point out in the Conclusion, our axiomatization approach can be easily modified in order to capture the definition of Monderer and Samet.

\smallskip

\smallskip

If $(M,s,v)\models \varphi$ holds for every valuation $v$ we write $(M,s)\models \varphi$. If $(M,s)\models \varphi$ for all $s \in S$, we write $M \models \varphi$. 


\begin{definition}[Satisfiability of sentences]
	A sentence $\varphi$ is \emph{satisfiable} if there is a state $s$ in some model $M$ such that $(M,s)\models \varphi$.  A set of sentences $T$ is satisfiable if there exists a state $s$ in a model $M$ such that $(M,s)\models \varphi$ for each $\varphi \in T$. A sentence $\varphi$ is \emph{valid}, if $\neg\varphi$ is not satisfiable. 
\end{definition}

Note that in the previous definition the satisfiability of sentences doesn't depend on a valuation, since they ton't contain any free variable.

In order to keep the satisfiability relation well-defined, here we consider only the  models in which all the sets of the form $$[\varphi] _{i, s}^v = \{ s \in S_{i, s} \, | \, (M,s,v) \models \varphi\},$$ are measurable.

\begin{definition}[Measurable model]
	A model $M=(S,D,I,\mathcal  K, \mathcal  P)$ is a \emph{measurable models} if $$[\varphi] _{i, s}^v  \in \chi_{i, s}, $$ for every formula $\varphi$, valuation $v$, state $s$ and agent $i$. We denote the class of all these models as $\mathcal M_{\mathcal A}^{MEAS}$. 
\end{definition}

Observe that if $\varphi$ is a sentence then the set $[\varphi] _{i, s}^v$ doesn't depend on $v$, thus we relax the notation by denoting  it by $[\varphi] _{i, s}$. Also, we write $\mu_{i, s} ([\varphi])$ instead of $\mu_{i, s} ([\varphi]_{i, s})$.

\subsection{Axiomatization issues}\label{sec ax issues} 

At the end of this section we analyze two common characteristics of epistemic logics and probability logics, which have impacts on their axiomatizations.

The first one is the \emph{non-compactness} phenomena -- there are unsatisfiable sets of formulas such that all their finite subsets are satisfiable.
The existence of such sets in epistemic logic is a consequence of the fact that the common knowledge operator $C_G$ can be semantically seen as an infinite conjunction of all the degrees of the group knowledge operator $E_G$, which leads to the example $$\{(E_G )^m \varphi \, | \, m\in \mathbb N \} \cup \{ \neg C_G \varphi  \}.$$ 
In real-valued probability logics, a standard example of unsatisfiable set whose finite subsets are all satisfiable is
$$\{  P_{i,\geq 1 - \frac{1}{n} }\varphi \, | \, m\in \mathbb N   \} \cup \{ \neg P_{i, \not = 1} \varphi  \}, $$
where $\varphi$ is a satisfiable sentence which is not valid.
One significant consequence of non-compactness is that there is no finitary axiomatization which is strongly complete \cite{vdH97}, i.e., simple completeness is the most one can  achieve.

In the first order case, situation is even worse. Namely, the set of valid formulas is not \emph{recursively enumerable}, neither for first order logic with common knowledge \cite{Wolter} nor for first order probability logics \cite{AH94} (moreover, even their monadic fragments suffer from the same drawback \cite{DBLP:books/sp/OgnjanovicRM16,Wolter}). This means that there is no finitary axiomatization which could be (even simply) complete. 
An approach for overcoming  this issue, proposed by Wolter \cite{Wolter}, is to consider infinitary logics as the only interesting alternative.

In this paper, we  introduce the axiomatization with \emph{$\omega$-rules} (inference rules with countably many premises) \cite{PLogics,Kooi}. This allows us to keep the object language  countable, and to move infinity to meta language only: the  formulas are finite, while only proofs are allowed to be infinite.

\section{The axiomatization $Ax_{PCK^{fo}}$}\label{section ax}

In this section we introduce the axiomatic system for the logic $PCK^{fo}$, denoted by $Ax_{PCK^{fo}}$. It  consists of the following axiom schemata and rules of inference:\\

I First-order axioms and rules\\

\hspace{4ex}Prop.  All instances of tautologies of the propositional calculus\\

\hspace{4ex}MP. $  \dfrac{\varphi, \varphi \to \psi }{\psi }$ (Modus Ponens)\\

\hspace{4ex}FO1. $\forall x (\varphi \to \psi) \to (\varphi \to \forall x \psi ) $, where $x$ is not a free variable un $\varphi$ \label{FO1}

\hspace{4ex}FO2. $\forall \varphi(x) \to \varphi(t)$, where $\varphi (t)$ is the result of substitution of all free occurences of $x$ in $\varphi(x)$ 

\hspace{9ex} by a term $t$ which is free for $x$ in $\varphi(x)$ \label{FO2}

\hspace{4ex}FO3. $\forall x K_i \varphi(x) \to K_i \forall x  \varphi(x)$ (Barcan formula)\label{FO3}

\hspace{4ex}FOR.  $ \dfrac{ \varphi}{ \forall x \varphi }$\label{FOR}  \\

II Axioms and rules for reasoning about knowledge\\

\hspace{4ex}AK.  $(K_i \varphi \wedge K_i(\varphi \to \psi) )\to  K_i\psi$, $i \in G$\label{AK} (Distribution Axiom)

\hspace{4ex}RK. $ \dfrac{ \varphi}{ K_i \varphi }$\label{RK}  (Knowledge Necessitation)

\hspace{4ex}AE. $E_G\varphi \to K_i \varphi$, $ i \in G$\label{AE} 

\hspace{4ex}RE.  $  \dfrac{ \{ \Phi_{k, \boldsymbol{\uptheta}, \mathbf{X}}(K_i \varphi) \, | \,  i \in G\} }{   \Phi_{k, \boldsymbol{\uptheta}, \mathbf{X}}(E_G \varphi)}$\label{RE}\\

\hspace{4ex}AC. $C_G\varphi \to (E_G)^m\varphi, m \in \mathbb{N}$ \label{AC}

\hspace{4ex}RC.  $  \dfrac{ \{ \Phi_{k, \boldsymbol{\uptheta}, \mathbf{X}}((E_G)^m \varphi)\, | \,  m \in \mathbb N\} }{  \Phi_{k, \boldsymbol{\uptheta}, \mathbf{X}}(C_G \varphi)}$\label{RC}\\

III Axioms and rule for reasoning about probabilities\\

\hspace{4ex}P1. $P_{i,\geq 0}{\varphi}$\label{P1}

\hspace{4ex}P2. $P_{i,\leq r}{\varphi} \to P_{i,< t}{\varphi}$, $t>r$\label{P2}

\hspace{4ex}P3. $P_{i,< t}{\varphi} \to P_{i,\leq t}{\varphi}$\label{P3}

\hspace{4ex}P4. $(P_{i,\geq r}{\varphi} \wedge P_{i,\geq t}{\psi} \wedge P_{i, \geq 1}{\neg (\varphi \wedge \psi) }) \to P_{i,\geq min(1, r+t)}(\varphi \vee \psi) $\label{P4}

\hspace{4ex}P5. $(P_{i,\leq r}{\varphi} \wedge P_{i,< t}{\varphi}) \to P_{i,< r+ t}{(\varphi \vee \psi)} $, $r+t \leq 1$\label{P5}

\hspace{4ex}RP. $ \dfrac{ \varphi}{ P_{i, \geq 1}{\varphi} }$  \label{RP}
(Probabilistic Necessitation)

\hspace{4ex}RA.  $  \dfrac{ \{ \Phi_{k, \boldsymbol{\uptheta}, \mathbf{X}}(P_{i,\geq r - \frac{1}{m}}{\varphi})\, | \, \mbox{$m \geq \frac{1}{r}$, $m \in \mathbb N$ }\} }{  \Phi_{k, \boldsymbol{\uptheta}, \mathbf{X}}(P_{i,\geq r}{\varphi} )}$, $r \in (0,1]_{\mathbb Q}$ (Archimedean rule)\label{RA}\\


IV Axioms and rules for reasoning about probabilistic knowledge\\

\hspace{4ex}APE. $E_G^r\varphi \to K_i^r \varphi$, $ i \in G$\label{APE}

\hspace{4ex}RPE.  $  \dfrac{ \{ \Phi_{k, \boldsymbol{\uptheta}, \mathbf{X}}(K_i^r \varphi) \, | \,  i \in G\} }{   \Phi_{k, \boldsymbol{\uptheta}, \mathbf{X}}(E_G^r \varphi)}$\label{RPE}\\

\hspace{4ex}APC. $C_G^r\varphi \to (F_G^r)^m\varphi,\, m \in \mathbb N $ \label{APC}

\hspace{4ex}RPC.  $  \dfrac{ \{ \Phi_{k, \boldsymbol{\uptheta}, \mathbf{X}}(F_G^r)^m \varphi)\, | \,  m \in \mathbb N\} }{  \Phi_{k, \boldsymbol{\uptheta}, \mathbf{X}}(C_G^r \varphi)}$\label{RPC}\\

The given axioms and rules are divided in four groups, according to the type of reasoning.
The first group contains the standard axiomatization for first-order logic and, in addition, a  variant of the well-known axiom for modal logics, called Barcan formula. It is proved that Barcan formula holds in the class of all first-order fixed domain modal models,
and that it is independent from the other  modal axioms 
\cite{hughes1968introduction,hughes1984companion}.
The second group contains axioms and rules for epistemic reasoning. AK and RK are classical Distribution axiom and Necessitation rule for the knowledge operator. The axiom AE and the rule RE are novel; they properly relate the knowledge operators and the operator of group knowledge $E_G$, regardless of the cardinality  of the group $G$. Similarly, AC and RC properly relate the operators $E_G$ and $C_G$. The infinitary rule RC is a generalization of the rule $InfC$ from \cite{Kooi}.
The third group contains multi-agent variant of a standard axiomatization for reasoning about probability \cite{DBLP:books/sp/OgnjanovicRM16}. The infinitary rule RA is a variant of so called Archimedean rule, generalized by incorporating the $k$-nested implications, in a similar way as it has been done in \cite{Milos} in purely probabilistic settings. This rule informally says that if probability of a formula is considered by an agent $i$ to be arbitrary close to some number $r$, then, according to the agent $i$, the probability of the fomula must be equal to $r$. 
The last group consist of novel axioms and rules which allow reasoning about probabilistic knowledge. They properly capture the semantic relationship between the operators $K_i^r$, $E_G^r$, $F_G^r$ and $C_G^r$, and they are in spirit similar to the last four axioms and rules from the second group.

Note that we use the structure of these $k$-nested implications in all of our infinitary inference rules. As we have already mentioned, the reason is that this form allows us to prove  Deduction theorem  and Strong necessitation theorem.
Note that by choosing $k=0$, $\theta_0=\top$ in the inference rules RE, RC, RPE, RPC, we obtain the intuitive forms of the rules:\\

$  \dfrac{ \{ K_i \varphi,  \, | \,  i \in G\} }{  E_G \varphi},  \dfrac{ \{ (E_G)^m \varphi\, | \, \forall m \geq 1\} }{  C_G \varphi}, \dfrac{ \{ K_i^r \varphi\, | \,  i \in G\} }{  E_G^r \varphi}, \dfrac{ \{ (F_G^r)^m \varphi, \, | \, \forall  m \geq 0\} }{  C_G^r \varphi}$.\\

Next  we define some basic notions of proof theory.

\begin{definition}
	
	A formula $\varphi $ is a \emph{theorem}, denoted by $\vdash \varphi$, if there is an at most countable sequence of formulas $\varphi_0, \varphi_1, \ldots , \varphi_{\lambda+1}$ ($\lambda$ is a finite or countable ordinal\footnote{Ie. the length of a proof is an at most countable \emph{successor} ordinal.})
	of
	formulas from $For_{PCK^{fo}} $, such that  $\varphi_{\lambda+1}=\varphi$, and
	every $\varphi_i$ is an instance of some axiom schemata or is obtained from the preceding formulas by an inference rule.  
	
	A formula $\varphi$ is \emph{derivable from} a set $T$ of formulas ($T \vdash \varphi$) if there is an at most countable sequence of formulas $\varphi_0, \varphi_1, \ldots , \varphi_{\lambda+1}$ ($\lambda$ is a finite or countable ordinal)  such that  $\varphi_{\lambda+1}=\varphi$, and each $\varphi_i$ is an instance of some axiom schemata or a formula from the set $T$, or it is obtained from the previous formulas by an inference rule, with the exception that the premises of the inference rules RK and RP  must be theorems. 
	The corresponding sequence of formulas is a \textit{proof} for $\varphi$ from $T$.
	
	A set of formulas $T$ is \emph{deductively closed} if it contains all the formulas derivable from $T$, i.e.,  $\varphi \in T$ whenever $T \vdash \varphi$.
	
\end{definition}

Obviously, a formula is a theorem iff it is derivable from the empty set. Now we introduce the notions of consistency and maximal consistency.

\begin{definition}
	A set $T$ of formulas is \emph{inconsistent} if $T \vdash \varphi$ for every formula $\varphi$, otherwise it is \emph{consistent}.  A set $T$ of formulas is \textit{maximal consistent} if it is consistent, and each proper superset of $T$ is inconsistent.
	
\end{definition}

It is easy to see that $T$ is inconsistent iff $T \vdash \bot$.

In the proof of completeness theorem, we will use a special type of maximal consistent sets, called saturated sets.

\begin{definition}
	
	A  set $T$ of formulas  is \emph{saturated} iff it is maximal consistent and the following condition holds:
	
	\begin{itemize}
		\item[] \quad \quad \quad \quad
		if $\neg(\forall x) \varphi(x) \in T$, then  there is a term $t$ such that $\neg\varphi(t) \in T$. \end{itemize}
	
\end{definition}

Note the notions of \emph{deductive closeness, maximal consistency} and \emph{saturates sets} are defined for formulas, but they can be defined for \emph{theories} (sets of sentences) in the same way.  We omit the formal definitions  here, since they would have the identical form as the ones above, but we will use the mentioned notions in the following sections.

\smallskip





The following result shows that the proposed  axioms from $Ax_{PCK^{fo}}$ are valid, and the inference rules preserve validity.

\begin{theorem}[Soundness] The axiomatic system $Ax_{PCK^{fo}}$ is sound with respect to the class of $PCK^{fo}$ models. 
\end{theorem}

\begin{proof}

	The soundness of the propositional part follows directly from the fact that interpretation of $\wedge$ and $\neg$  in the definition of $\models$ relation  is the same as in the propositional calculus. 
	The proofs for FO1. and FOR. are standard.
	
	\smallskip
	
	

	AE. AC. and APC.\, follow immediately from the semantics of operators $E_G$, $C_G$ and $C_G^r$. 
	
	\smallskip
	
	FO2. Let $(M,s) \models (\forall x) \varphi(x)$. Then $(M,s,v) \models (\forall x) \varphi(x)$ for every valuation $v$. Note that for every   $v$, among all valuations there must be a valuation $v'$ such that $v'(s)(x) = d = I(s)(t)_v$ and $(M,s,v') \models \varphi(x)$.
	From the equivalence  $(M,s,v') \models \varphi(x)$ iff $(M,s,v) \models \varphi(t)$, we obtain that $(M,s,v) \models \varphi(t)$ holds for every valuation. Thus, every instance of FO2 is valid.
	
	\smallskip
	
	FO3. (Barcan formula) Suppose that $(M,s) \models (\forall x) K_i \varphi (x)$, ie. for each evaluation $v$, $(M,s, v) \models (\forall x)  K_i \varphi (x)$. Then for each valuation $v$ and every $d \in D$, $  (M,s,v[d/x])\models K_i \varphi (x)$. Therefore for every $v$ and $d$ and every $t \in \mathcal K_i(s)$, we have $(M,t,v[d/x]) \models \varphi(x)$. Thus, for every $t \in K_i(s)$, and every valuation $v$, $(M,t, v) \models (\forall x) \varphi (x)$. Finally, since for every $t \in K_i(s)$, $(M,t) \models (\forall x)  \varphi (x)$, we have $(M,s) \models K_i (\forall x)   \varphi (x)$.
	
	\smallskip

	

	RC. We will prove by induction on $k$ that if $(M,s, v) \models  \Phi_{k, \boldsymbol{\uptheta}, \mathbf{X}}((E_G)^m  \varphi)$, for all $m\in \mathbb N$, then also $(M,s, v) \models \Phi_{k, \boldsymbol{\uptheta}, \mathbf{X}}(C_G \varphi)$, for each state $s$ and valuation $v$ of any Kripke structure $M$:
	
	Induction base $k=0$. Let $(M,s, v) \models \theta_0 \to (E_G)^m  \varphi$, $ \mbox{ for all } m\in \mathbb{N} $. Assume that it is not $(M,s, v) \models \theta_0 \to  C_G \varphi$, i.e., \begin{equation}\label{sj}
		(M,s, v) \models \theta_0 \wedge \neg C_G \varphi.
	\end{equation}
	Then $(M,s, v) \models (E_G)^m  \varphi$, $ \mbox{ for all } m\in \mathbb{N} $, and therefore $(M,s, v) \models C_G \varphi$ (by the definition of the satisfiability relation), which contradicts (\ref{sj}).


	Inductive step. Let  $(M,s, v) \models  \Phi_{k+1, \boldsymbol{\uptheta}, \mathbf{X}}((E_G)^m  \varphi)$, $\mbox{ for all } m\in \mathbb{N} $.
	
	Suppose $X_{k+1}=K_i$ for some $i \in \mathcal A$ ie. $(M,s, v) \models  \theta_{k+1} \to K_i\Phi_{k, \boldsymbol{\uptheta}, \mathbf{X}}((E_G)^m  \varphi), \mbox{ for all } m\in \mathbb{N}$. 
	Assume the opposite, that $(M,s, v) \not \models  \Phi_{k+1, \boldsymbol{\uptheta}, \mathbf{X}}(C_G \varphi)$, ie. $(M,s, v) \models  \theta_{k+1} \wedge \neg K_i\Phi_{k, \boldsymbol{\uptheta}, \mathbf{X}}(C_G \varphi)$. Then also $(M,s, v) \models  K_i\Phi_{k, \boldsymbol{\uptheta}, \mathbf{X}}((E_G)^m  \varphi), \mbox{ for all } m\in \mathbb{N}$, so for every state $t \in \mathcal K_i(s)$ we have that $(M,t, v) \models  \Phi_{k, \boldsymbol{\uptheta}, \mathbf{X}}((E_G)^m  \varphi), \mbox{ for all } m\in \mathbb{N}$, and by the induction hypothesis $(M,t, v) \models  \Phi_{k, \boldsymbol{\uptheta}, \mathbf{X}}(C_G \varphi)$. Therefore  $(M,s, v) \models K_i  \Phi_{k, \boldsymbol{\uptheta}, \mathbf{X}}(C_G \varphi)$, leading to a contradiction.

	On the other hand, let $X_{k+1}=P_{i,\geq 1}, i\in \mathcal A$ i.e.  $(M,s, v) \models \theta_{k+1} \to P_{i,\geq 1} \Phi_{k, \boldsymbol{\uptheta}, \mathbf{X}}((E_G)^m  \varphi),\mbox{ for all } m\in \mathbb{N}$. Otherwise, if $(M,s, v) \not \models \Phi_{k+1, \boldsymbol{\uptheta}, \mathbf{X}}(C_G \varphi)$, then $(M,s, v) \models \theta_{k+1} \wedge \neg P_{i,\geq 1}\Phi_{k, \boldsymbol{\uptheta}, \mathbf{X}}(C_G \varphi)$, so $(M,s, v) \models P_{i,\geq 1} \Phi_{k, \boldsymbol{\uptheta}, \mathbf{X}}((E_G)^m  \varphi)$ for every $m \in \mathbb N$, $m \geq \frac{1}{r}$. This implies there is a subset $U \subseteq S_{i,s}$ such that $\mu_{i,s}(U)=1$ and for all $u \in U$, $m \in \mathbb N$, $m \geq \frac{1}{r}$:  $(M,u, v) \models  \Phi_{k, \boldsymbol{\uptheta}, \mathbf{X}}((E_G)^m  \varphi )$. Then $(M,u, v) \models  \Phi_{k, \boldsymbol{\uptheta}, \mathbf{X}}(C_G \varphi)$ for all $u \in U$ by the induction hypothesis, so $(M,s, v) \models P_{i,\geq 1} \Phi_{k, \boldsymbol{\uptheta}, \mathbf{X}}(C_G \varphi )$, which is a contradiction. 
	
	\smallskip
	
	RA. We prove the soundness of this rule by induction on $k$, ie. if $(M,s, v) \models \Phi_{k, \boldsymbol{\uptheta}, \mathbf{X}}(P_{i,\geq r - \frac{1}{m}}{\varphi} )$  for every $m \in \mathbb N$, $m \geq \frac{1}{r}$ and $r>0$, given some model $M$, state $s$ and valuation $v$, then $(M,s, v) \models \Phi_{k, \boldsymbol{\uptheta}, \mathbf{X}}(P_{i,\geq r}{\varphi})$.
	
	Induction base $k=0$. This case follows by the properties of the real numbers.
	
	Inductive step. Let $(M,s, v) \models \Phi_{k+1, \boldsymbol{\uptheta}, \mathbf{X}}(P_{i,\geq r - \frac{1}{m}}{\varphi})$ and $X_{k+1}=P_{i,\geq 1}, i\in \mathcal A$ ie. $(M,s, v) \models \theta_{k+1} \to P_{i,\geq 1} \Phi_{k, \boldsymbol{\uptheta}, \mathbf{X}}(P_{i,\geq r - \frac{1}{m}}{\varphi})$  for every $m \in \mathbb N$, $m \geq \frac{1}{r}$. Assume the opposite, that $(M,s, v) \not \models \Phi_{k+1, \boldsymbol{\uptheta}, \mathbf{X}}(P_{i,\geq r}{\varphi})$. Then $(M,s, v) \models \theta_{k+1} \wedge \neg P_{i,\geq 1}\Phi_{k, \boldsymbol{\uptheta}, \mathbf{X}}(P_{i,\geq r}{\varphi})$, so $(M,s, v) \models P_{i,\geq 1} \Phi_{k, \boldsymbol{\uptheta}, \mathbf{X}}(P_{i,\geq r - \frac{1}{m}}{\varphi} )$ for every $m \in \mathbb N$, $m \geq \frac{1}{r}$. Therefore, there exists a subset $U \subseteq S_{i,s}$ such that $\mu_{i,s}(U)=1$ and for all $u \in U$, $m \in \mathbb N$, $m \geq \frac{1}{r}$:  $(M,u, v) \models  \Phi_{k, \boldsymbol{\uptheta}, \mathbf{X}}(P_{i,\geq r - \frac{1}{m}}{\varphi} )$. Then, by the induction hypothesis, we have $(M,u, v) \models  \Phi_{k, \boldsymbol{\uptheta}, \mathbf{X}}(P_{i,\geq r})$ for all $u \in U$, so $(M,s, v) \models P_{i,\geq 1} \Phi_{k, \boldsymbol{\uptheta}, \mathbf{X}}(P_{i,\geq r}{\varphi} )$, which is a contradiction. The proof follows similarly as for RC if $X_{k+1}=K_i, i \in \mathcal A$.

	\smallskip

	
	RPC. Now we show that rule RPC preserves validity by induction on $k$. 
	
	Let us prove the implication: if $(M,s,v) \models \Phi_{k, \boldsymbol{\uptheta}, \mathbf{X}}((F_G^r)^m \varphi)$, for all $m\in \mathbb N$, and $PCK^{fo}_\infty$-models $M$, then also $(M,s,v) \models \Phi_{k, \boldsymbol{\uptheta}, \mathbf{X}}(C_G^r \varphi)$, for each state $s$ in $M$:
	
	Induction base $k=0$. Suppose $(M,s,v) \models \theta_0 \to (F_G^r)^m \varphi$ for all $m\in \mathbb N$. If it is not $(M,s,v) \models \theta_0 \to  C_G^r \varphi$, i.e. $(M,s,v) \models \theta_0 \wedge \neg C_G^r \varphi$, then $(M,s,v) \models (F_G^r)^m \varphi$, for all $m \in \mathbb N$. Therefore $(M,s,v) \models C_G^r \varphi$, which is a contradiction.
	
	Inductive step. Let  $(M,s,v) \models  \Phi_{k+1, \boldsymbol{\uptheta}, \mathbf{X}}((F_G^r)^m \varphi), \mbox{ for all }  m\in \mathbb{N}\ $.

	Suppose $X_{k+1}=K_{i}, i\in \mathcal A$ i.e. $(M,s,v) \models \theta_{k+1} \to K_i\Phi_{k, \boldsymbol{\uptheta}, \mathbf{X}}((F_G^r)^m \varphi), \mbox{ for all }  m\in \mathbb N$. If $s \not \models  \Phi_{k+1, \boldsymbol{\uptheta}, \mathbf{X}}(C_G^r \varphi)$, i.e. $(M,s,v) \models  \theta_{k+1} \wedge \neg K_i\Phi_{k, \boldsymbol{\uptheta}, \mathbf{X}}(C_G^r \varphi)$  (*), then $(M,s,v) \models  K_i\Phi_{k, \boldsymbol{\uptheta}, \mathbf{X}}((F_G^r)^m \varphi)$, for all  $m\in \mathbb N$. So for each $t \in \mathcal K_i(s)$ we have $(M,t,v) \models \Phi_{k, \boldsymbol{\uptheta}, \mathbf{X}}((F_G^r)^m \varphi)$. By the induction hypothesis on $k$ it follows that $(M,t,v) \models \Phi_{k, \boldsymbol{\uptheta}, \mathbf{X}}(C_G^r \varphi)$. But then $(M,s,v) \models K_i  \Phi_{k, \boldsymbol{\uptheta}, \mathbf{X}}(C_G^r \varphi)$ which contradicts (*).
	
	Let $X_{k+1}=P_{i,\geq 1}, i\in \mathcal A$ i.e.  $(M,s, v) \models \theta_{k+1} \to P_{i,\geq 1} \Phi_{k, \boldsymbol{\uptheta}, \mathbf{X}}((F_G^r)^m \varphi),\mbox{ for all } m\in \mathbb{N}$. Otherwise, if $(M,s, v) \not \models \Phi_{k+1, \boldsymbol{\uptheta}, \mathbf{X}}(C_G \varphi)$, then $(M,s, v) \models \theta_{k+1} \wedge \neg P_{i,\geq 1}\Phi_{k, \boldsymbol{\uptheta}, \mathbf{X}}(C_G^r\varphi)$, so $(M,s, v) \models P_{i,\geq 1} \Phi_{k, \boldsymbol{\uptheta}, \mathbf{X}}((F_G^r)^m \varphi)$ for every $m \in \mathbb N$, $m \geq \frac{1}{r}$. Therefore, there is a subset $U \subseteq S_{i,s}$ such that $\mu_{i,s}(U)=1$ and for all $u \in U$, $m \in \mathbb N$, $m \geq \frac{1}{r}$:  $(M,u, v) \models  \Phi_{k, \boldsymbol{\uptheta}, \mathbf{X}}((F_G^r)^m  \varphi )$. Then $(M,u, v) \models  \Phi_{k, \boldsymbol{\uptheta}, \mathbf{X}}(C_G^r \varphi)$ for all $u \in U$ by the induction hypothesis, so $(M,s, v) \models P_{i,\geq 1} \Phi_{k, \boldsymbol{\uptheta}, \mathbf{X}}(C_G^r \varphi )$ which is a contradiction. 
\end{proof}

\section{Some theorems of $PCK^{fo}$}

In this section we prove several theorems. 
Some of them will
be useful  in proving the completeness of the axiomatization $Ax_{PCK^{fo}}$. 
We start with the deduction theorem. Since we will frequently use this theorem, we will not always  explicitly mention it in the proofs.

\begin{theorem}[Deduction theorem]\label{deduc}  If $T$ is a theory and $\varphi, \psi$ are sentences, then $T \cup \{ \varphi \} \vdash \psi$ implies $T \vdash \varphi \to \psi$.
\end{theorem}

\begin{proof}
	
	We use the transfinite induction on the length of the proof of $\psi$ from $T \cup \{\varphi\}$. The case $\psi = \varphi$ is obvious; if $\psi$ is an axiom, then $\vdash \psi$, so $T \vdash \psi$, and therefore $T \vdash \varphi \to \psi$. If $\psi$ was obtained by rule RK, ie. $\psi = K_i \varphi$ where $\varphi$ is a theorem, then $\vdash K_i \varphi$ (by R2), that is, $\vdash \psi$, so $T \vdash \varphi \to \psi$. The reasoning is analogous for cases of other inference rules that require a theorem as a premise. Now we consider the case where $\psi$ was obtained by rule RPC. The proof for the other infinitary rules is similar.\vspace{2pt}
	
	\smallskip
	
	\smallskip
	
	Let $T, \varphi \vdash \{ \Phi_{k, \boldsymbol{\uptheta}, \mathbf{X}}((F_G^r)^m \eta)\, | \,  m \in \mathbb N\} \vdash \psi$ where  $\psi =\Phi_{k, \boldsymbol{\uptheta}, \mathbf{X}}(C_G^r \eta), k \geq 1$. Then 
	
	$T \vdash \varphi \to  \Phi_{k, \boldsymbol{\uptheta}, \mathbf{X}}((F_G^r)^m \eta)$, $\mbox{ for all }  m\in \mathbb N  $, by the induction hypothesis.\vspace{2pt}
	
	Suppose $X_{k}=K_i$, for some $i \in \mathcal A$ . \vspace{2pt}
	
	$T \vdash \varphi \to   ( \theta_k \to   K_i  \Phi_{k-1, \boldsymbol{\uptheta}, \mathbf{X}}((F_G^r)^m \eta))$, by the definition of $\Phi_{k}$\vspace{2pt}

	$T \vdash  (\varphi \wedge \theta_k) \to    K_i  \Phi_{k-1, \boldsymbol{\uptheta}, \mathbf{X}}((F_G^r)^m \eta)$, by the propositional tautology
	
	\ \ \ \ \ \ $(p \to (q \to r)) \longleftrightarrow ((p \wedge q) \to r)$.\vspace{2pt}
	
	Let  $\overline{\boldsymbol{\uptheta}} = (\theta_0, \dots, \theta_{k-1}, \varphi \wedge \theta_k)$. Then we have:\vspace{2pt}
	
	$T \vdash \overline{ \theta_k} \to     K_i\Phi_{k-1, \boldsymbol{\uptheta}, \mathbf{X}}((F_G^r)^m \eta),  \mbox{ for all }  m\in \mathbb N $\vspace{2pt}
	
	$T \vdash  \Phi_{k, \boldsymbol{\uptheta}, \mathbf{X}}((F_G^r)^m \eta), \mbox{ for all }  m\in \mathbb N  $\vspace{2pt}
	
	$T \vdash  \Phi_{k, \boldsymbol{\uptheta}, \mathbf{X}}(C_G^r \eta) $ by RPC\vspace{2pt}
	
	$T \vdash  (\varphi \wedge \theta_k) \to   K_i  \Phi_{k-1, \boldsymbol{\uptheta}, \mathbf{X}}(C_G^r\eta) $\vspace{2pt}
	
	$T \vdash \varphi \to   ( \theta_k \to   K_i  \Phi_{k-1, \boldsymbol{\uptheta}, \mathbf{X}}(C_G^r\eta) $\vspace{2pt}
	
	$T \vdash \varphi \to  \Phi_{k, \boldsymbol{\uptheta}, \mathbf{X}}(C_G^r \eta) $ \vspace{2pt}
	
	$T \vdash \varphi \to  \psi $. 
	
	\smallskip
	
	\smallskip
	
	The case when $X_{k}=P_{i,\geq 1}$, for some $ i\in \mathcal A$ can be proved in the same way, by replacing $K_i$ with $P_{i,\geq 1}$. The case $k=0$ also follows in a similar way.
\end{proof}

Next we prove several results about purely epistemic part of our logic. First we show  that the strong variant of necessitation for knowledge operator is a consequence of the axiomatization $Ax_{PCK^{fo}}$. This theorem will have an important role in the proof of completeness theorem, in the construction of the canonical model.

First we need to introduce some notation. For a given set of formulas $T$ and $ i \in \mathcal A$, we define the set $K_i T$ as the set of all formulas $K_i\varphi$, where $\varphi$ belongs to $T$, i.e.
$$K_iT=\{K_i\varphi \, | \, \varphi \in T \}.$$

\begin{theorem}[Strong necessitation]\label{thm strong necessitation}If $T$ is a theory and $T \vdash \varphi$, then $K_iT \vdash K_i  \varphi$, for all $i \in \mathcal A.$
\end{theorem}

\begin{proof}
	Let $T \vdash \varphi$.
	We will prove $K_iT \vdash K_i  \varphi$ using the transfinite induction on the length of proof of $T \vdash \varphi$.
	Here we will only consider the application of rules FOR and RPC, while the cases when we apply the other infinitary rules are similar as  the proof for RCP.
	
	\smallskip

	1) Suppose that $T \vdash  \varphi$, where $\varphi = (\forall x) \psi$, was obtained from $T\vdash \psi$ by the inference rule FOR. 
	
	Then:
	
	\smallskip
	
	$ T\vdash \psi $ by the assumption
	
	$K_i T\vdash K_i \psi$ by the induction hypothesis

	$K_i T\vdash  (\forall x) K_i \psi$ by FOR
	
	$K_i T\vdash   K_i (\forall x) \psi$ by Barcan formula
	
	\smallskip
	
	2) Suppose that  $T \vdash  \varphi$ where  $\varphi = \Phi_{k, \boldsymbol{\uptheta}, \mathbf{X}}(C_G^r \psi)$ was derived by application of RPC. Then:
	
	\smallskip
	
	
	$T \vdash\Phi_{k, \boldsymbol{\uptheta}, \mathbf{X}}((F_G^r)^m\psi )$,  for all $ m \in \mathbb N$
	
	$K_iT \vdash K_i\Phi_{k, \boldsymbol{\uptheta}, \mathbf{X}}((F_G^r)^m\psi )$, for all $ m \in \mathbb N$,   by induction hypothesis
	
	$K_iT \vdash \top \to K_i\Phi_{k, \boldsymbol{\uptheta}, \mathbf{X}}((F_G^r)^m\psi )$,  for all $ m \in \mathbb N$ 
	
	$K_iT \vdash \Phi_{k+1, \overline{\boldsymbol{\uptheta}}, \overline{\mathbf{X}} }((F_G^r)^m\psi)$, where $\overline{\boldsymbol{\uptheta}} = (\boldsymbol{\uptheta}, \top) $   and $ \overline{\mathbf{X}} = (\mathbf{X}, K_i)$. 
	
	$K_iT \vdash \Phi_{k+1, \overline{\boldsymbol{\uptheta}}, \overline{\mathbf{X}} }(C_G^r \psi)$, by RPC
	
	$K_iT \vdash \top \to K_i  \Phi_{k, \boldsymbol{\uptheta}, \mathbf{X}}(C_G^r \psi)$
	
	$K_iT \vdash \top \to K_i  \varphi$
	
	$K_iT \vdash K_i  \varphi$.  
\end{proof}

As a consequence, we also obtain strong necessitation for the operators of group knowledge. As we will see later, this result is necessary to prove so-called fixed-point axiom for common knowledge operator.

\begin{corollary}\label{cor necessitation}
	If $T$  \label{necess} is a theory and $T \vdash \varphi$, then $E_G T\vdash E_G \varphi$, for all $ G \subseteq \mathcal A$.
\end{corollary}

\begin{proof}
	Let $T \vdash \varphi$. For every $i \in G$, we have  $E_GT \vdash K_i T$ by the axiom AE, and $K_iT \vdash K_i \varphi$, by Theorem \ref{thm strong necessitation}. Since by the rule RE,  where we choose $k=0$ and $\theta_0=\top$, we have
	$$
	\{  K_i \varphi \ | \ i\in G \}\vdash E_G \varphi,
	$$
	we obtain $E_GT \vdash E_G \varphi$.
\end{proof}


Now we show that some standard properties of epistemic operators can be proved in $Ax_{PCK^{fo}}$.

\begin{proposition}\label{prop epistemic}
	Let $\varphi$, $\psi$, $\varphi_j$, $j=1,\ldots,m$ be formulas, $i\in \mathcal{A}$ and $G \in \mathcal G$. Then:
	
	\begin{enumerate}
		
		\item \label{KD} $\vdash K_i(\varphi \to \psi) \to (K_i\varphi \to K_i\psi)$
		
		
		\item \label{EMD} $\vdash E_G(\varphi \to \psi) \to (E_G\varphi \to E_G\psi)$

		\item \label{CD} $\vdash C_G(\varphi \to \psi) \to (C_G\varphi \to C_G\psi)$


		
		
		
		
		\item \label{conjunction k}  $\vdash K_i(\bigwedge\limits_{j=1}^m \varphi_j)\equiv\bigwedge\limits_{j=1}^m K_i\varphi_j,  \forall i \in G$,  
		
		
		\item \label{conjunction e} $\vdash E_G(\bigwedge\limits_{j=1}^m \varphi_j)\equiv\bigwedge\limits_{j=1}^m E_G\varphi_j$
		
		\item \label{axiom fixed point} $\vdash C_G\varphi \to E_G(\varphi \wedge C_G \varphi) $ 
		

	\end{enumerate}
\end{proposition}



\begin{proof}\hfill
	
	
	\begin{enumerate}
		
		\item follows directly from AK. 

		\item We use the following derivation:
		\begin{align*}
			E_G \varphi \wedge E_G(\varphi \to \psi) &\vdash \{ K_i \varphi \wedge K_i(\varphi \to \psi) \, | \, \forall i \in G\} \mbox{ (by AE)} \\
			&\vdash \{ K_i \psi \, | \, \forall i \in G\} \mbox{ (by AK)}\\
			&\vdash E_G \psi \mbox{ (by RE)}
		\end{align*}
		Therefore, by \hyperref[deduc]{Deduction theorem} $\vdash E_G \varphi \wedge E_G(\varphi \to \psi) \to  E_G\psi$, i.e.,  $\vdash E_G(\varphi \to \psi) \to (E_G\varphi \to E_G\psi)$.

		\item Let us first prove, using the induction on $n$, that
		\begin{equation}\label{n1}	
			\vdash (E_G)^m  (\varphi \to \psi) \to ((E_G)^m  \varphi \to (E_G)^m  \psi)
		\end{equation}
		holds for every $m \in \mathbb{N}$.
		
		Induction base is proved in the previous part of this proposition (\ref{EMD}).
		
		Induction step:
		
		$\vdash (E_G)^m  (\varphi \to \psi) \to ((E_G)^m  \varphi \to (E_G)^m  \psi)$, induction hypothesis
		
		$\vdash K_i ( (E_G)^m  (\varphi \to \psi) \to ((E_G)^m  \varphi \to (E_G)^m  \psi)),  \forall i \in G$, by RK 
		
		$\vdash E_G ( (E_G)^m  (\varphi \to \psi) \to ((E_G)^m  \varphi \to (E_G)^m  \psi))$, by RE
		
		$\vdash E_G ( (E_G)^m  (\varphi \to \psi) \to ((E_G)^m  \varphi \to (E_G)^m  \psi))    \to (     E_G^{m+1}(\varphi \to \psi)        \to          E_G((E_G)^m  \varphi \to (E_G)^m  \psi)               )$, by induction base
		
		$\vdash     (E_G)^{m+1}(\varphi \to \psi)        \to          E_G((E_G)^m  \varphi \to (E_G)^m  \psi)                             $, by previous two
		
		$\vdash  E_G((E_G)^m  \varphi \to (E_G)^m  \psi)       \to ((E_G)^{m+1}\varphi \to (E_G)^{m+1}\psi)                      $, by induction base
		
		$\vdash (E_G)^{m+1}(\varphi \to \psi) \to ((E_G)^{m+1}\varphi \to (E_G)^{m+1}\psi)$, by previous two.
		
		Thus, (\ref{n1}) holds. Next,	\begin{align*}
			C_G \varphi \wedge C_G(\varphi \to \psi) &\vdash \{ (E_G)^m  \varphi \wedge (E_G)^m (\varphi \to \psi) \, | \, \forall m \in \mathbb{N}\} \mbox{ (by AC)} \\
			&\vdash \{ (E_G)^m  \psi \, | \, \forall m \in \mathbb{N}\} \mbox{ (by (\ref{n1}))}\\
			&\vdash C_G \psi \mbox{ (by RC)}
		\end{align*}	
		Then $\vdash C_G(\varphi \to \psi) \to (C_G\varphi \to C_G\psi)$, by \hyperref[deduc]{Deduction theorem}.

		\item This standard result in modal logics follows from Distribution axiom and propositional reasoning.
		
		\item First we prove that $E_G(\bigwedge\limits_{j=1}^m \varphi_j)$ implies $ \bigwedge\limits_{j=1}^m E_G \varphi_j$.
		\begin{align*}
			E_G(\bigwedge\limits_{j=1}^m \varphi_j) &\vdash \{ K_i( \bigwedge\limits_{j=1}^m \varphi_j) \,| \, i \in G\}  \mbox{ (by AE)}\\
			&\vdash  \{  \bigwedge\limits_{j=1}^m K_i  \varphi_j \, | \, i \in G \}   \mbox{ (by the  previous part of the proposition (\ref{conjunction k}))} \\
			&\vdash  \bigcup\limits_{j=1}^m \{  K_i  \varphi_j \, | \, i \in G \}  \mbox{ (since } \bigwedge\limits_{j=1}^m K_i  \varphi_j \to K_i  \varphi_j, \,\, \forall j=1,...,m )\\
			&\vdash   \bigcup\limits_{j=1}^m \{ E_G \varphi_j\, | \, i \in G \}  \mbox{ (by RE)}\\
			&\vdash  \bigwedge\limits_{j=1}^m E_G \varphi_j \,\,\, \mbox{ (by propositional reasoning)}
		\end{align*}
		Conversely,
		\begin{align*}
			\bigwedge\limits_{j=1}^m E_G\varphi_j &\vdash \{ K_i \varphi_1 \,|\, i \in G  \}  \cup  \{ K_i \varphi_2 \,|\, i \in G  \} \cup ... \cup \{ K_i \varphi_m \,|\, i \in G  \} \mbox{ (by AE) }\\
			&\vdash  \{  \bigwedge\limits_{j=1}^m K_i  \varphi_j \, | \, i \in G \}   \\
			&\vdash \{ K_i( \bigwedge\limits_{j=1}^m \varphi_j) \,| \, i \in G\}   \mbox{ (by the  previous part of the proposition (\ref{conjunction k}))} \\
			&\vdash E_G(\bigwedge\limits_{j=1}^m \varphi_j)  \,\,\,  \mbox{ (by RE)}
		\end{align*}
		
		Therefore, by \hyperref[deduc]{Deduction theorem} we have that   $\vdash E_G(\bigwedge\limits_{j=1}^m \varphi_j)\equiv\bigwedge\limits_{j=1}^m E_G\varphi_j$.

		\item   
		
		$\vdash C_G\varphi \to E_G \{  (E_G)^m \varphi \, | \, m \in \mathbb N    \} $, by AC
		
		$ E_G \{  (E_G)^m \varphi \, | \, m \in \mathbb N    \} \vdash E_G C_G \varphi$, by RC and  Corollary  \ref{cor necessitation}

		$\vdash C_G\varphi \to E_G  C_G \varphi$, by previous two
		
		$\vdash C_G\varphi \to E_G\varphi $, by AC
		
		$\vdash C_G\varphi \to E_G(\varphi \wedge C_G \varphi) $, by previous two and the previous part (\ref{conjunction e}) of the proposition.
	\end{enumerate}
\end{proof}

Note that  
(\ref{CD})
and
(\ref{axiom fixed point}) (the fixed-point axiom) are two standard axioms of epistemic logic with common knowledge \cite{KP,Guide}. The axiom (\ref{CD}) is often written in an equivalent form $$(C_G\varphi \wedge C_G(\varphi\to \psi))\to C_G \psi.$$ The previous result shows that they are provable in our axiomatic system $Ax_{PCK^{fo}}$.

\smallskip

The standard axiomatization for epistemic logics (with finitely many agents) \cite{KP,Guide} also includes one axiom for group knowledge operator, which states that group knowledge $E_G\varphi$ is equivalent to the conjunction of $K_i\varphi$, where all  the agents $i$ from the group are considered. The next result shows that both that axiom and its probabilistic variant hold in our logic.

\begin{proposition}
	Let $\varphi$ be a formula, $r\in [0,1]_\mathbb Q$, and let $G \in \mathcal G$ be a finite set of agents. Then the following hold.
	\begin{enumerate}
		
		\item  $\vdash E_G \varphi \equiv \bigwedge_{i\in G}K_i\varphi$
		\item  $\vdash E^r_G \varphi \equiv \bigwedge_{i\in G}K^r_i\varphi$
	\end{enumerate}  
	
\end{proposition}

\begin{proof}\hfill
	\begin{enumerate}
		
		\item From the axiom AE, using propositional reasoning, we can obtain $\vdash E_G \varphi \to \bigwedge_{i\in G}K_i\varphi$. On the other hand, from the inference rule RE, choosing $k=0$ and $\theta_0=\top$, we obtain 
		$ \{ K_i \varphi  \, | \,  i \in G\}\vdash  E_G \varphi$, i.e., $\bigwedge_{i\in G}K_i\varphi\vdash  E_G \varphi$, so $\vdash\bigwedge_{i\in G}K_i\varphi\to  E_G \varphi$ follows from Deduction theorem.
		
		\item This result can be proved in the same way as the first statement, using the obvious analogies between the axioms AE and APE, and the rules RE and RPE.
	\end{enumerate}  
\end{proof}

Note that the distribution properties of the epistemic operators $K_i$, $E_G$ and $C_G$, proved in Proposition \ref{prop epistemic} (1)-(3), cannot be directly transferred to the properties of the corresponding operators of probabilistic knowledge. For example, it is easy to see that  $ E^r_G(\varphi \to \psi) \to (E^r_G\varphi \to E^r_G\psi)$ is not a valid formula.\footnote{On the other hand, it can be shown that the formula $ E^1_G(\varphi \to \psi) \to (E^r_G\varphi \to E^r_G\psi)$ is valid and it is a theorem of our logic (see (\ref{E proof})).} Nevertheless, we can prove that probabilistic versions of  knowledge, group knowledge and common knowledge are closed under consequences.

\begin{proposition}
	Let $\varphi$ and $\psi$  be formulas such that $\vdash\varphi\to\psi$. Let $r\in[0,1]_\mathbb Q$, $i\in \mathcal{A}$ and $G \in \mathcal G$. Then:

	\begin{enumerate}


		\item  \label{KPD1} $\vdash K^r_i\varphi \to K^r_i\psi $ 
		
		\item  \label{EPD1} $\vdash E_G^r\varphi \to E_G^r\psi$ 
		
		
		\item \label{CD1}  $ \vdash C_G^r\varphi \to C_G^r\psi$

	\end{enumerate}
\end{proposition}



\begin{proof}\hfill
	
	
	\begin{enumerate}

		\item 	Note that 	
		\begin{equation}\label{eqn}
			\vdash K_i(P_{i,\geq 1}{(\varphi \to \psi)} \to (  P_{i,\geq r}{\varphi }  \to P_{i,\geq r}{\psi}) ) \to (K_iP_{i,\geq 1}{(\varphi \to \psi)} \to K_i(  P_{i,\geq r}{\varphi }  \to P_{i,\geq r}{\psi}) )
		\end{equation}
		by Proposition \ref{prop epistemic}(\ref{KD}).
		From the assumption $\vdash \varphi\to \psi$, applying the rule RP and then the rule RK, we obtain 
		\begin{equation}\label{eq1n}
			\vdash K_i^1(\varphi\to \psi).
		\end{equation}
		
		Note that $\vdash \neg \varphi \lor \neg \bot$ (a propositional tautology), so	
		\begin{equation}\label{eq1}
			\vdash P_{i,\geq 1}( \neg \varphi \lor \neg \bot), \mbox{ by RP }
		\end{equation}
		
		Also, $\vdash \neg (\varphi \land \neg \bot) \lor \neg \neg \varphi $, so
		\begin{equation}\label{eq2}	
			\vdash P_{i,\geq 1}(\neg (\varphi \land \neg \bot) \lor \neg \neg \varphi),  \mbox{ by  RP }
		\end{equation}

		By P4 we have $\vdash (P_{i,\geq r}{\varphi} \wedge P_{i,\geq 0}{\neg \bot} \wedge P_{i, \geq 1}{(\neg \varphi \lor \neg \bot) }) \to P_{i,\geq 1}(\varphi \vee \bot) $, so	
		\begin{equation}\label{eq3}	
			\vdash P_{i,\geq r}{\varphi} \to P_{i,\geq r}{(\varphi \lor \bot)},\mbox{ by (\ref{eq1}) using the instance $P_{i,\geq 0}{\neg \bot}$ of P1 } 
		\end{equation}
		\noindent The formula $P_{i,\geq r}{(\varphi \lor \bot)}$ denotes $P_{i,\geq r}{\neg (\neg \varphi \land \neg \bot)}$, which is the same as $P_{i,\geq 1-(1-r)}{\neg (\neg \varphi \land \neg \bot)}$, and can be abbreviated as $P_{i,\leq 1-r}{(\neg \varphi \land \neg \bot)}$. Similarly, $\neg P_{i,\geq r} \neg \neg \varphi$ denotes $P_{i,< r} \neg \neg \varphi$. From P5 we obtain 
		$\vdash (P_{i,\leq 1-r}{(\neg \varphi \land \neg \bot)} \wedge P_{i,< r}{\neg \neg \varphi}) \to P_{i,< 1}{((\neg \varphi \land \neg \bot) \vee \neg \neg \varphi)} $.
		
		Since $P_{i,\geq 1}(\neg (\varphi \land \neg \bot) \lor \neg \neg \varphi)$ denotes $\neg P_{i,< 1}{((\neg \varphi \land \neg \bot) \vee \neg \neg \varphi)} $, from (\ref{eq2}) we have
		
		$\vdash (P_{i,\leq 1-r}{(\neg \varphi \land \neg \bot)} \wedge P_{i,< r}{\neg \neg \varphi}) \to
		%
		P_{i,< 1}{((\neg \varphi \land \neg \bot) \vee \neg \neg \varphi)} \land \neg P_{i,< 1}{((\neg \varphi \land \neg \bot) \vee \neg \neg \varphi))} $, by P5,
		and therefore
		$\vdash P_{i,\leq 1-r}{(\neg \varphi \land \neg \bot)} \to P_{i,< r}{\neg \neg \varphi} $, i.e.,	
		\begin{equation}\label{eq4}	
		\vdash P_{i,\geq r}{( \varphi \lor  \bot)} \to P_{i,\geq r}{\neg \neg \varphi} 
		\end{equation}
		
		From (\ref{eq3}) and (\ref{eq4}) we obtain 	$\vdash P_{i,\geq r}{( \varphi )} \to P_{i,\geq r}{\neg \neg \varphi} $. The negation of the formula 
		\begin{equation}\label{prob k}	
			P_{i,\geq 1}{(\varphi \to \psi)} \to (  P_{i,\geq r}{\varphi }  \to P_{i,\geq r}{\psi})
		\end{equation}
		is equivalent to $P_{i,\geq 1}(\neg \varphi \lor \psi) \land P_{i,\geq r}\varphi \land P_{i,<r}\psi$. Since $P_{i,\geq r}{ \varphi } \to P_{i,\geq r}{\neg \neg \varphi} $, then $P_{i,\geq 1}(\neg \varphi \lor \psi) \land P_{i,\geq r}{\neg \neg \varphi} \land P_{i,<r}\psi$, which can be written as $P_{i,\geq 1}(\neg \varphi \lor \psi) \land P_{i,\leq 1-r}{\neg  \varphi} \land P_{i,<r}\psi$. Then
		$\vdash P_{i,\leq 1-r}{\neg  \varphi} \land P_{i,<r}\psi \to P_{i,< r}(\neg \varphi \lor \psi)$, by P5, and since
		$ P_{i,<1}{ \varphi}$ is an abbreviation for $\neg P_{i,\geq 1}{ \varphi} $, we have
		$\vdash \neg( P_{i,\geq 1}{(\varphi \to \psi)} \to (  P_{i,\geq r}{\varphi }  \to P_{i,\geq r}{\psi})) \to P_{i,\geq 1}(\neg \varphi \lor \psi) \land \neg P_{i,\geq 1}(\neg \varphi \lor \psi)$, a contradiction.	Thus, the formula (\ref{prob k}) is a theorem of our axiomatization. By applying the rule RK to the theorem, we obtain 
		\begin{equation}\label{prob k 2}\vdash K_i (P_{i,\geq 1}{(\varphi \to \psi)} \to (  P_{i,\geq r}{\varphi }  \to P_{i,\geq r}{\psi}) )
		\end{equation}
		
		From (\ref{eqn}) and (\ref{prob k 2}) we obtain
		\begin{equation}\label{prob k 3}\vdash K_iP_{i,\geq 1}{(\varphi \to \psi)} \to K_i(  P_{i,\geq r}{\varphi }  \to P_{i,\geq r}{\psi}).
		\end{equation}
		By Proposition \ref{prop epistemic}(\ref{KD}), we have
		\begin{equation}\label{prob k 4}\vdash K_i(  P_{i,\geq r}{\varphi }  \to P_{i,\geq r}{\psi}) \to ( K_i P_{i,\geq r}{\varphi } \to  K_i  P_{i,\geq r}{\psi }  ).
		\end{equation}
		From (\ref{prob k 3}) and (\ref{prob k 4}), we obtain 
		$\vdash K_iP_{i,\geq 1}{(\varphi \to \psi)} \to	( K_i P_{i,\geq r}{\varphi } \to  K_i  P_{i,\geq r}{\psi }  ) $, i.e.,
		\begin{equation}\label{prob k 5}\vdash K^1_i(\varphi \to \psi) \to (K^r_i\varphi \to K^r_i\psi).
		\end{equation}
		Finally, from (\ref{eq1n}) and (\ref{prob k 5}) we obtain $\vdash K^r_i\varphi \to K^r_i\psi$.

		\item We start with the following derivation:	\begin{align*}
			E_G^r \varphi \wedge E_G^1(\varphi \to \psi) &\vdash \{ K_i^r \varphi \wedge K_i^1(\varphi \to \psi) \, | \, \forall i \in G\} \mbox{, by APE} \\
			&\vdash \{ K_i^r \psi \, | \, \forall i \in G\} \mbox{, by (\ref{prob k 5})}\\
			&\vdash E_G^r \psi \mbox{, by RPE}
		\end{align*}
		Therefore, 
		\begin{equation}\label{E proof}\vdash E_G^1(\varphi \to \psi) \to (E_G^r\varphi \to E_G^r\psi)
		\end{equation} by \hyperref[deduc]{Deduction theorem}.
		From (\ref{eq1}), using the rule RPE we obtain 
		\begin{equation}\label{E1}\vdash E_G^1(\varphi \to \psi).
		\end{equation}
		Finally, from (\ref{E proof}) and (\ref{E1}) we obtain $\vdash E_G^r\varphi \to E_G^r\psi$.

		\item First we prove that
		\begin{equation}\label{C1}
			\vdash (F_G^r)^{m}\varphi \to (F_G^r)^{m}\psi
		\end{equation}
		holds for every $m$.
		We prove the claim by induction.
		
		Induction base follows trivially since $ (F_G^r)^0 \varphi = \top$.
		
		Suppose that $\vdash (F_G^r)^m\varphi \to (F_G^r)^m\psi$ (induction hypothesis).
		
		$\vdash (\varphi \wedge (F_G^r)^m\varphi) \to (\varphi \wedge (F_G^r)^m\psi)$
		
		$\vdash P_{i,\geq 1}((\varphi \wedge (F_G^r)^m\varphi) \to (\varphi \wedge (F_G^r)^m\psi)), \forall i \in G$, by RP 
		
		$\vdash K_iP_{i,\geq 1}((\varphi \wedge (F_G^r)^m\varphi) \to (\varphi \wedge (F_G^r)^m\psi)), \forall i \in G$, by RK 
		
		$\vdash E_G^1((\varphi \wedge (F_G^r)^m\varphi) \to (\varphi \wedge (F_G^r)^m\psi))$, by RPE
		
		$\vdash E_G^1((\varphi \wedge (F_G^r)^m\varphi) \to (\varphi \wedge (F_G^r)^m\psi)) \to (E_G^r(\varphi \wedge (F_G^r)^m\varphi)     \to E_G^r (\varphi \wedge (F_G^r)^m\psi) )$, by (\ref{E proof}) 
		
		$\vdash E_G^r(\varphi \wedge (F_G^r)^m\varphi)     \to E_G^r (\varphi \wedge (F_G^r)^m\psi) $ by previous two, ie.
		
		$\vdash (F_G^r)^{m+1}\varphi \to (F_G^r)^{m+1}\psi$.
		
		Thus, (\ref{C1}) holds.
		\begin{align*}
			C_G^r \varphi  &\vdash \{ (F_G^r)^m \varphi \, | \, \forall m \in \mathbb{N}_0\} \mbox{ (by APC)} \\
			&\vdash \{ (F_G^r)^m \psi \, | \, \forall m \in \mathbb{N}_0\} \mbox{ , by (\ref{C1}) }\\
			&\vdash C_G^r \psi \mbox{ , by RPC}
		\end{align*}
		Now $\vdash C_G^r \varphi \to C_G^r \psi$ follows from Deduction theorem.
		
	\end{enumerate}
\end{proof}

At the end of this section, we prove several results about maximal consistent sets with respect to our axiomatic system. Those results will be useful in proving the Truth lemma.
\begin{lemma}\label{maxconsprop}
	
	Let $T$ be a maximal consistent set of formulas for $Ax_{PCK^{fo}}$. Then $T$ satisfies the following properties:
	
	\begin{enumerate}

		\item  for every formula $\varphi$, exactly one of $\varphi$ and $\neg \varphi$ is in $T$,
		
		\item \label{ded clo} $T$ is deductively closed,
		
		\item $\varphi \land \psi \in T$ iff $\varphi \in T$ and $\psi \in T$,
		
		
		\item \label{imp} if $\{\varphi, \varphi \to \psi\} \subseteq T$, then $\psi \in T$,
		
		
		\item if $r = \sup \, \{q  \in [0,1]_{\mathbb Q} \, | \, P_{i,\geq q}{\varphi} \in T\}$ and $r \in [0,1]_{\mathbb Q}$, then $P_{i,\geq r}{\varphi} \in T$.
		
	\end{enumerate}
	
\end{lemma} 

\begin{proof}\hfill
	
	\begin{enumerate}
		
		\item  If both formulas $\varphi, \neg \varphi \in T$, $T$ would be inconsistent. Suppose $\varphi \not \in T$. Since $T$ is maximal, $T \cup \{\varphi\}$ is inconsistent, and by the \hyperref[deduc]{Deduction theorem} $T \vdash \neg \varphi$. Similarly, if $\neg \varphi \not \in T$, then $T \vdash \varphi$. Therefore, if both formulas $\varphi, \neg \varphi \not \in T$, set $T$ would be inconsistent, so exactly one of them is in $T$.
		
		\item Otherwise, if there is some $\varphi$ such that $T\vdash \varphi$ and $\varphi \not \in T$ then, by the previous part of this lemma, $\neg \varphi \in T$, so $T$ would be inconsistent. 
		
		\item  Suppose $\varphi \in T$ and $\psi \in T$. Then $T \vdash \varphi$, $T \vdash \psi$, $T \vdash \varphi \land \psi$ and $ \varphi \land \psi \in T$, because $T$ is deductively closed by Lemma \ref{maxconsprop}(\ref{ded clo}). For the other direction, let $\varphi \land \psi \in T$. Then $T \vdash \varphi \land \psi$, $T \vdash (\varphi \land \psi) \to \varphi$, $T \vdash (\varphi \land \psi) \to \psi$, $T \vdash \varphi$ and $T \vdash \psi$. Therefore $\varphi, \psi \in T$, by Lemma \ref{maxconsprop}(\ref{ded clo}).
		
		
		\item If $\{\varphi, \varphi \to \psi\} \subseteq T$, then $T \vdash \varphi$, $T \vdash \varphi \to \psi$ and $T \vdash \psi$, so $\psi \in T$ by Lemma \ref{maxconsprop}(\ref{ded clo}).
		
		
		\item Let $r = \sup \, \{q  \, | \, P_{i,\geq q}{\varphi} \in T\}$, thus  $T \vdash P_{i,\geq q}{\varphi}$ for every $q<r$, $q \in [0,1]_{\mathbb Q}$. Then by the Archimedean rule RA, we have that $T \vdash P_{i,\geq r}{\varphi}$. Therefore $P_{i,\geq r}{\varphi} \in T$ by Lemma \ref{maxconsprop}(\ref{ded clo}).
	\end{enumerate}
\end{proof}

\begin{lemma} \label{max}
	
	Let $V$ be a maximal consistent set of formulas.
	
	\begin{enumerate}
		
		\item $E_G \varphi \in V$ iff ( $K_i \varphi \in V$ $\, \mbox{ for all }  i \in G$ )\vspace{2pt}
		
		\item $E_G^r \varphi \in V$ iff ( $K_i^r \varphi \in V$ $\, \mbox{ for all }  i \in G$ )\vspace{2pt}
		
		\item $C_G \varphi \in V$ iff ( $(E_G)^m \varphi \in V$ $\, \mbox{ for all }  m \in \mathbb N$ )\vspace{2pt}
		
		\item $C_G^r \varphi \in V$ iff ($ (F_G^r)^m \varphi \in V$  $\, \mbox{ for all }  m \in \mathbb N $)\vspace{2pt}
		
	\end{enumerate}

\end{lemma}

\begin{proof}
	
	For the proof of (1), suppose that $E_G \varphi \in V$. Since  $E_G\varphi \to K_i \varphi$, $ \mbox{for all } i \in G$ is the axiom AE, then also $E_G\varphi \to K_i \varphi \in V$  $\mbox{ for all }  i \in G $. Therefore $ K_i \varphi \in V$  $ \mbox{for all }  i \in G $  by Lemma \ref{maxconsprop}(\ref{imp}) because $V$ is maximal consistent. For the other direction, if $ K_i \varphi \in V$  $\mbox{ for all }  i \in G $, and since $   \{  K_i \varphi \ | \ i\in G \}\vdash E_G \varphi$ (by the rule RE, where $k=0$ and $\theta_0=\top$), we have that $E_G\varphi \in V$, by Lemma \ref{maxconsprop}(\ref{ded clo}).
	
	\smallskip
	
	The cases (2), (3) and (4) can be proved in a similar way, by replacing  $E_G \varphi$, $K_i \varphi$, $ \mbox{for all } i \in G$, axiom AE and rule RE with $E_G^r \varphi, K_i^r \varphi$, $ \mbox{for all } i \in G$, APE, RPE (case (2)), $C_G \varphi, (E_G)^m \varphi, \mbox{ for all }  m \in \mathbb N$, AC, RC (case (3)), and $C_G^r \varphi, (F_G^r)^m \varphi \mbox{ for all }  m \in \mathbb N$, APC, RPC (case (4)), respectively.
\end{proof}

\section{Completeness}

In this section we prove that the axiomatic system $Ax_{PCK^{fo}}$ is strongly complete for the class of measurable $\mathcal{M}^{MEAS}_\mathcal{A}$ models, using a Henkin-style construction \cite{henkin1949}. We prove completeness in three steps. First, we extend a theory $T$ to a saturated theory $T^*$ step by step, in an infinite process, considering in each step one sentence and checking its consistency with the considered theory in that step. Due to the presence of infinitary rules, we modify the standard completion technique in the case that the considered sentence can be derived by an infinitary rules, by adding the negation of one of the premisses of the rule. Second, we use the saturated theories to construct a special $PCK^{fo}$ model, that we will call \emph{canonical model}, and we show that it belongs to the class $\mathcal{M}^{MEAS}_\mathcal{A}$.
Finally, using the saturation $T^*$ of the considered theory $T$, we show that $T$ is satisfiable in the corresponding state $s_{T^{*}}$ of the canonical model.

\subsection{Lindenbaum's theorem}\label{sect lindenbaum}

We start with the Henkin construction of saturated extensions of theories. For that purpose, we consider a broader language, obtained by adding countably many novel constant symbols.

\begin{theorem}[Lindenbaum's theorem]\label{lindenbaum} Let $T$ be a consistent theory in the language $\mathcal L_{PCK^{fo}}$, and $C$ an infinite enumerable set of new constant symbols (i.e. $C \cap \mathcal L_{PCK^{fo}} = \emptyset$).  Then $T$ can be extended to a saturated theory $T^{*}$ in the language $\mathcal L^*= L_{PCK^{fo}} \cup C$. 
\end{theorem}

\begin{proof}
	
	Let $\{ \varphi_i \, | \, i \in \mathbb N\}$ be an enumeration of all sentences in $Sent_{PCK^{fo}}$. Let $C$ be an infinite enumerable set of constant symbols such that $C \cap \mathcal L_{PCK^{fo}} = \emptyset$. We define the family  of  theories $(T_i)_{i\in \mathbb N}$, and the set $T^*$ in the following way:\\

	1. $T_0=T$.
	
	2.  For every  $i \in \mathbb N$:
	
	\enspace  a. if $T_i \cup \{ \varphi _i \}$ is consistent, then $T_{i+1}=T_i \cup \{ \varphi_i\}$
	
	\enspace  b. if $T_i \cup \{ \varphi _i \}$ is inconsistent, and
	
	\quad b1.  $\varphi_i= \Phi_{k, \boldsymbol{\uptheta}, \mathbf{X}}(E_G \varphi) $, then $T_{i+1}=T_i \cup \{\neg \varphi_i, \neg  \Phi_{k, \boldsymbol{\uptheta}, \mathbf{X}}(K_j \varphi) \}$, 
	
	\hspace{25pt}for some $j \in G$ such that $T_{i+1}$ is consistent
	
	\quad b2.  $\varphi_i=\Phi_{k, \boldsymbol{\uptheta}, \mathbf{X}}(C_G \varphi)$, then $T_{i+1}=T_i \cup \{\neg \varphi_i, \neg  \Phi_{k, \boldsymbol{\uptheta}, \mathbf{X}}((E_G)^m  \varphi) \}$,
	
	\hspace{25pt}for some $m \in \mathbb N$ such that $T_{i+1}$ is consistent
	
	\quad b3.  $\varphi_i= \Phi_{k, \boldsymbol{\uptheta}, \mathbf{X}}(E_G^r \varphi) $, then $T_{i+1}=T_i \cup \{\neg \varphi_i, \neg  \Phi_{k, \boldsymbol{\uptheta}, \mathbf{X}}(K_j^r \varphi) \}$,
	
	\hspace{25pt}for some $j \in G$ such that $T_{i+1}$ is consistent
	
	\quad  b4.  $\varphi_i=\Phi_{k, \boldsymbol{\uptheta}, \mathbf{X}}(C_G^r \varphi)$, then $T_{i+1}=T_i \cup  \{\neg \varphi_i, \neg  \Phi_{k, \boldsymbol{\uptheta}, \mathbf{X}}((F_G^r)^m \varphi) \}$, for some $m \in \mathbb N$ such that $T_{i+1}$ is consistent
	
	\quad b5.  $\varphi_i= \Phi_{k, \boldsymbol{\uptheta}, \mathbf{X}}(P_{i,\geq r }{\varphi})$, then $T_{i+1}=T_i \cup \{\neg \varphi_i, \neg  \Phi_{k, \boldsymbol{\uptheta}, \mathbf{X}}(P_{i,\geq r - \frac{1}{m}}{\varphi})
	\}, $ for some $m \in \mathbb N$ such that $T_{i+1}$ is consistent
	
	\quad b6.  $\varphi_i = (\forall x) \varphi(x)$, then $T_{i+1}=T_i \cup \{\neg \varphi_i, \neg \varphi(c)\}$  for some constant symbol $c \in C$ which doesn't occur in any of the formulas from $T_i$ such that $T_{i+1}$ remains consistent

	\enspace c. Otherwise, $T_{i+1}=T_i \cup \{\neg \varphi_i \}$.

	3. $T^* = \bigcup\limits_{i=0}^{\infty}  T_i$.\\
	
	First we need to prove that the set $T^*$ is well defined, i.e. we need to show that the agents $j \in G$ used the steps b1. and b3. exist, that the numbers $m \in \mathbb N$ used in the steps b2., b4. and b5. exist, and that the constant $c \in C$ from the step b6. exists.  Let us prove correctness in step b4. exists, i.e.,  that if $T_i \cup \{ \Phi_{k, \boldsymbol{\uptheta}, \mathbf{X}}(C_G^r \varphi)\}$ is inconsistent, then there exists $m \geq 1$ such that $T_i \cup \{ \neg \Phi_{k, \boldsymbol{\uptheta}, \mathbf{X}}((F_G^r)^m \varphi)\}$ is consistent. Otherwise, if  $T_i \cup \{\neg \Phi_{k, \boldsymbol{\uptheta}, \mathbf{X}}((F_G^b)^m \varphi)\}$ would be inconsistent for every $m$, then $T_i \vdash  \Phi_{k, \boldsymbol{\uptheta}, \mathbf{X}}((F_G^r)^m \varphi)$ for each $m$ by \hyperref[deduc]{Deduction theorem}, and therefore $T_i \vdash \Phi_{k, \boldsymbol{\uptheta}, \mathbf{X}}(C_G^r \varphi)$ by the inference rule RPC. But since $T_i \cup \{ \Phi_{k, \boldsymbol{\uptheta}, \mathbf{X}}(C_G^r \varphi)\}$ is inconsistent, we have $T_i \vdash \neg \Phi_{k, \boldsymbol{\uptheta}, \mathbf{X}}(C_G^r \varphi)$, which is in a contradiction with consistency of $T_i$. In a similar way we can prove existence of  $j$ and $m$ in the steps b1-b5., where the other infinitary rules are considered. Let us now consider the case b6. It is obvious that the formula  $\neg (\forall x) \varphi(x)$ can be consistently added  to $T_{i}$, and if there is already some $c \in C$ such that $\neg \varphi(c) \in T_i$, the proof is finished.
	If there is no such $c$, observe that $T_{i}$ is constructed by adding finitely many formulas to $T$, so  there is a constant symbol $c \in C$ which does not appear in $T_{i} $. Let us show that we can choose that $c$ in b6. If we suppose that $T_{i} \cup \{\neg (\forall x)\varphi(x), \neg \beta(c) \} \vdash \bot$, then by Deduction theorem we have 
	$T_i, \neg (\forall x)\varphi(x) \vdash \varphi(c)$.
	Note that $c$ does not appear in  $T_{i} \cup \{\neg (\forall x) \varphi(x) \}$, and therefore
	$T_i, \neg (\forall x) \varphi(x) \vdash (\forall x) \varphi(x)$, which is impossible.
	Thus, the sets $T_i$ are well defined. Note that they are consistent by construction.

	Next we prove that $T^*$   is  deductively closed, using the induction on the length of proof. The proof is straightforward in the case of finitary rules. Here we will only prove that $T^*$ is closed under the rule  RPC,   since  the cases when other infinitary rules are considered can be treated  in a similar way. 
	
	Suppose $T^{*} \vdash \phi$ was obtained by RPC, where $\Phi_{k, \boldsymbol{\uptheta}, \mathbf{X}}((F_G^r)^n \varphi) \in T^*$ for all $n \in \mathbb N$, and $\phi = \Phi_{k, \boldsymbol{\uptheta}, \mathbf{X}}(C_G^r \varphi)$.  Assume  that $\Phi_{k, \boldsymbol{\uptheta}, \mathbf{X}}(C_G^r \varphi ) \not \in T^*$. 
	Let $i$ be the positive integer such that $\varphi_i=\Phi_{k, \boldsymbol{\uptheta}, \mathbf{X}}(C_G^r \varphi)$. Then $T_i \cup \{\varphi_i\} $ is inconsistent, since otherwise $\Phi_{k, \boldsymbol{\uptheta}, \mathbf{X}}(C_G^r \varphi)=\varphi_i \in T_{i+1}\subset T^*$. Therefore $T_{i+1}=T_i \cup \{  \neg  
	\Phi_{k, \boldsymbol{\uptheta}, \mathbf{X}}((F_G^r)^m \varphi)  \}$ for some $m$, so $\neg \Phi_{k, \boldsymbol{\uptheta}, \mathbf{X}}((F_G^r)^m \varphi) \in  T^*$, which contradicts the consistency of $T_j$.
	
	If we would suppose $T$ is inconsistent, i.e.  $T^*\vdash \bot$, then we would have $\bot\in T^*$ since $T^*$ is deductively closed. Therefore, there would be some  $i$ such that $\bot\in T_i$, which is impossible. Thus, $T^*$ is consistent. 
	
	Finally, the step b6. of the construction guaranties that the theory $T^*$ is saturated in the language $\mathcal L^*$.
\end{proof}

\subsection{Canonical model}\label{sec canonical}

Now we construct a special Kripke structure, whose set of states consists of saturated theories.
First we need to introduce some notation. For a given set of formulas $T$ and $ i \in \mathcal A$, we define the set $T/K_i$ as the set of all formulas $\varphi$, such $K_i\varphi$ belongs to $T$, i.e.
$$T/K_i=\{\varphi \, | \, K_i\varphi \in T \}.$$

\begin{definition}[Canonical model]\label{def canonical}
	The \emph{canonical model} is the structure $M^{*} = (S,D,I, \mathcal{K}, \mathcal{P})$, such that
	
	\begin{itemize}
		\item  $S = \{s_V \ |\ V \mbox{ is a saturated theory}\}$
		\item $D$ is the set of all variable-free terms 
		\item $\mathcal K_i = \{ (s_V, s_U)\ | \  V/K_i \subseteq U \}$,  $\mathcal K = \{ \mathcal K_i \, | \, i \in \mathcal A\}$
		
		\item $I(s)$  is an interpretation such that:
		
		\begin{itemize}
			\item  for each function symbol $f^k_j$, $I(s)(f^k_j)$ is a function from $D^k$ to $D$ such that for all variable-free terms $t_1, \cdots,t_k$, $I(s)(f^k_j):$$ (t_1, \cdots,t_k)$$ \to f^k_j(t_1, \cdots,t_k)$
			
			\item for each relational symbol $R^k_j$,\\ $I(s)(R^k_j) = \{ (t_1, \cdots,t_k) \ | $ $ \ t_1, \cdots,t_k $  are variable-free terms in $ R^k_j(t_1, \cdots,t_k) \in V,$ where $s=s_V \}$
		\end{itemize}
		
		\item $\mathcal P (i,s) = (S_{i, s}, \chi_{i, s}, \mu_{i, s})$, where  
		
		\begin{itemize}
			\item $S_{i, s} = S$
			\item  $\chi_{i, s} = \{[\varphi]_{i, s} \, | \, \varphi \in Sent_{PCK}\}$, where $[\varphi]_{i, s} = \{ s_V \in S_{i, s} \, | \, \varphi \in V \}$ 
			\item   if $ [\varphi]_{i, s} \in \chi_{i, s}$ then $\mu _{i, s} ([\varphi]_{i, s}) = \sup \, \{r  \, | \, P_{i,\geq r}{\varphi} \in V$, where $s=s_V\}$  
		\end{itemize}
		
	\end{itemize}
	
\end{definition}

Note that the sets $[\varphi]_{i, s}$ in the definition of the canonical mode actually don't depend on $i$ and $s$, so in the rest o this section  we will sometimes relax the notation by  omitting  the subscript.

Also, since there is a bijection between saturated theories and states of canonical model, we will often write just  $s$ when we denote either a state of the corresponding saturated theory. For example, we can write the last item of the definition above as  $\mu _{i, s} ([\varphi]_{i, s}) = \sup \, \{r  \, | \, P_{i,\geq r}{\varphi} \in s\}$.

\smallskip

Now we will show that $M^*$ is a
$PCK^{fo}$ model.
First we need show that each $\mathcal P (i,s)$ defines a is a probability space. 
In specific, we prove that the definition of $\mu _{i, s} $ is correct, i.e., that $\mu _{i, s} ([\varphi]_{i, s})$ doesn't depend on the way we choose a sentence from the class $[\varphi]_{i, s}$.

\begin{lemma}\label{lemaPres} 
	Let $M^{*} = (S,D,I, \mathcal{K}, \mathcal{P})$ be the canonical model. Then for each agent $i \in \mathcal A$ and $s \in S$ the following hold.\footnote{The proof of Lemma \ref{lemaPres} is essentially the same as the proofs of corresponding statements in single-agent probability logics \cite{DBLP:books/sp/OgnjanovicRM16}. We present it here for the completeness of the paper, and also because some steps in the proof will be useful for the proof of Theorem \ref{thm con}.}
	
	\begin{enumerate}
		
		\item If $\varphi$ and $\psi$ are two sentences such that $[\varphi]_{i, s}=[\psi]_{i, s}$, then 
		$ \sup \, \{r  \, | \, P_{i,\geq r}{\psi} \in s\}=\sup \, \{r  \, | \, P_{i,\geq r}{\varphi} \in s\}$
		
		\item $\mathcal P (i,s)$ is a is a probability space.
	\end{enumerate}	
\end{lemma}

\begin{proof}\hfill 	
	
	\begin{enumerate}
		
		\item If $[\varphi]_{i, s}=[\psi]_{i, s}$, then $\varphi$ and $\psi$ belong to the same saturated theories, so $\vdash \varphi\equiv \psi$. From $\vdash \varphi \to \psi$ we obtain $\vdash P_{i,\geq 1} (\varphi \to \psi)$ by RA, and therefore for every $r$ we have $\vdash P_{i,\geq r} \varphi \to P_{i,\geq r} \psi$ by (\ref{prob k}). Consequently, $P_{i,\geq r} \varphi \to P_{i,\geq r} \psi\in s$.
		If $P_{i,\geq r} \varphi \in s$ then, by Lemma \ref{maxconsprop}(\ref{imp}), also $P_{i,\geq r} \psi \in s$. Therefore, $ \sup \, \{r  \, | \, P_{i,\geq r}{\psi} \in s\}\geq \sup \, \{r  \, | \, P_{i,\geq r}{\varphi} \in s\}$. In the same way we can prove $ \sup \, \{r  \, | \, P_{i,\geq r}{\psi} \in s\}\leq \sup \, \{r  \, | \, P_{i,\geq r}{\varphi} \in s\}$ using 
		$\vdash \psi \to \varphi$.
		
		\item First we show that for each agent $i \in \mathcal A$ and $s \in S$, the class  $\chi_{i, s} = \{[\varphi] \, | \, \varphi \in Sent_{PCK_\infty}\}$ is an algebra of subsets of $S_{i, s}$.
		Obviously, we have that $S_{i, s} = [\varphi \lor \neg \varphi]$, for every formula $\varphi$.
		Also, if $[\varphi] \in \chi_{i, s}$, then $[\neg \varphi]$ is a complement of the set     $[\varphi]$, and it belongs to $\chi_{i, s}$
		Finally,  if $[\varphi_1], [\varphi_2] \in \chi_{i, s}$, then $[\varphi_1] \cup  [\varphi_2] \in \chi_{i, s}$ because $[\varphi_1] \cup  [\varphi_2] = [\varphi_1 \lor \varphi_2].$ 
		Therefore,  each  $\chi_{i, s}$ is an algebra of subsets of $S_{i, s}$.
		
		Note that from the axiom  $P_{i,\geq o} \varphi$ we can obtain  $\mu _{i, s}([\varphi]) \geq 0$. 
		Next we show $\mu _{i, s}([\varphi]) = 1- \mu _{i, s}([\neg \varphi])$. 
		Suppose $q = \mu _{i, s} ([\varphi]) = sup \, \{r  \, | \, P_{i,\geq r}{\varphi} \in s\}$. If $q=1$, then $P_{i,\geq r}{\varphi} = P_{i,\leq 0}\neg \varphi = \neg P_{i,>0} \neg \varphi$ and $\neg P_{i,>0} \neg \varphi \in s$. If for some $l>0$, $P_{i,\geq l} \neg \varphi \in s$ then $ P_{i,> 0}\neg \varphi \in s$, by axiom P2, which is a contradiction. Therefore, $\mu _{i, s} ([\varphi])=1$. Suppose $q<1$. Then for every rational number $q' \in (q,1]$, $\neg P_{i,\geq q'} \varphi = P_{i,< q'} \varphi$, so $P_{i,< q'} \varphi \in s$. Then by P2, $P_{i,\leq q'} \varphi$ and $P_{i,\geq 1-q'} \neg \varphi \in s$. On the other hand, if there is a rational $q'' \in [0,r)$ such that $P_{i,\geq 1-q''} \neg \varphi \in s$, then $\neg P_{i,>q''} \in s$, which is a contradiction. Therefore, $sup \, \{r  \, | \, P_{i,\geq r}{\neg \varphi} \in s\} = 1- sup \, \{r  \, | \, P_{i,\geq r}{\varphi} \in s\}$. Thus, $\mu _{i, s}([\varphi]) = 1- \mu _{i, s}([\neg \varphi])$.
		Let $[\varphi]_{i, s} \cap [\psi]_{i, s} = \emptyset$, $\mu _{i, s}([\varphi]) =q$, $\mu _{i, s}([\psi]) =l$. Since $[\psi]_{i, s} \subset [\neg \varphi]_{i, s}$, it follows that $q+l \leq q+(1-q)=1$. Suppose that $q,l>0$. Because of supremum and monotonicity properties, for all rational numbers $q'\in [0,q)$ and $l'\in [0,l)$: $P_{i,\geq q'} \varphi$, $P_{i,\geq l'} \psi \in s$. Then  $P_{i,\geq q'+l'} (\varphi \lor \psi) \in s$ by P4. Therefore, $q+l \leq sup \, \{r  \, | \, P_{i,\geq r}{(\varphi \lor \psi)} \in s\}$. If $q+l=1$, the statement is obviously valid. Suppose $q+l<1$. If  $q+l <r_0= \leq sup \, \{r  \, | \, P_{i,\geq r}{(\varphi \lor \psi)} \in s\}$, then for each rational $r' \in (q+l, r_0), P_{i,\geq r'}{(\varphi \lor \psi)} \in s $ . Let us choose rational $q''>q$ and $s''>s$ such that $\neg P_{i,\geq q''} \varphi, P_{i,< q''}\varphi \in s$, $\neg P_{i,\geq l''} \psi, P_{i,< l''}\psi \in s$ and $q''+l''=r'\leq1$. Then $P_{i, \leq q''}\varphi  \in s$ by the axiom P3. And by P5 we have $P_{i,\leq q''+'l''}(\varphi \lor \psi)$, $\neg P_{i,\geq q''+'l''}(\varphi \lor \psi)$ and $\neg P_{i,\geq r'}(\varphi \lor \psi)$, which is a contradiction. Therefore  $\mu _{i, s}([\varphi] \cup [\psi]) = \mu _{i, s}([\varphi]) + \mu _{i, s}([\psi])$. Finally, let us assume that $q=0$ or $l=0$. In that case we can repeat the previous reasoning, by taking either $q'=0$ or $l'=0$.
	\end{enumerate}
\end{proof}

The previous result still doesn't ensure that $M^*$  belongs to the class $\mathcal{M}^{MEAS}_\mathcal{A}$. Indeed, in Definition \ref{def canonical} the sets $[\varphi]$ from $\chi_{i,s}$  are defined using $\varphi\in T$, and not $(M^*,s_T)\models \varphi$. However, the following lemma shows that the former and later coincide.

\begin{lemma}[Truth lemma]\label{Truth} Let $T$ be a saturated theory. Then 
	$$\varphi \in T \ \  \mbox{ iff } \ \ (M^{*}, s_{T}) \models \varphi.$$\end{lemma}

\begin{proof} 	
	
	We prove the equivalence  by induction on complexity of $\varphi$:
	
	\smallskip
	
	- If the formula $\varphi$ is atomic, then  $\varphi \in T   \mbox{ iff }  (M^{*}, s_{T}) \models \varphi$, by the definition of $I(s)$ in $M^{*}$.
	
	- Let $\varphi = \neg \psi$. Then $ (M^{*}, s_{T}) \models \neg \psi$ iff $ (M^{*}, s_{T}) \not \models \psi$ iff $\psi \not \in T$ (induction hypothesis) iff $\neg \psi \in T$.
	
	- Let $\varphi = \psi \wedge \eta$. Then  $ (M^{*}, s_{T}) \models \psi \wedge \eta$ iff $ (M^{*}, s_{T}) \models \psi $ and $ (M^{*}, s_{T}) \models \eta$  iff $\psi \in T$ and $\eta \in T$  (induction hypothesis) iff  $\psi \wedge \eta \in T$ by Lemma \ref{maxconsprop}(3).
	
	- Let $\varphi = (\forall x) \psi$ and $\varphi \in T$. Then $\psi(t/x)$ for all $t \in D$ by FO2. It follows that $ (M^{*}, s_{T}) \models \psi(t/x) $ for all $t \in D$ by induction hypothesis, and therefore $(M^{*}, s_{T}) \models (\forall x) \psi$. In other direction, let $(M^{*}, s_{T}) \models (\forall x) \psi$ and assume the opposite ie. $\varphi = (\forall x) \psi \not \in T$. Then there exists some term $t \in D$ such that $ (M^{*}, s_{T}) \models \neg \psi(t/x) $ ( $T$ is saturated), leading to a contradiction $(M^{*}, s_{T}) \not \models (\forall x) \psi$.

	-Let $\varphi = P_{i,\geq r} \psi$. If $\varphi \in T$ then $\sup \, \{q  \, | \, P_{i,\geq q}{\psi} \in T\} = \mu _{i, s_T}([\psi]) \geq r$ and $ (M^{*}, s_{T}) \models P_{i,\geq r} \psi$. In other direction, let  $ (M^{*}, s_{T}) \models P_{i,\geq r} \psi$, i.e.,  $\sup \, \{q  \, | \, P_{i,\geq q}{\psi} \in T\}  \geq r$. If $\mu _{i, s_T}([\psi])>r$, then $P_{i,\geq r} \psi \in T$ because of the properties of supremum and monotonicity of the probability measure $\mu _{i, s_T}$.	 If $\mu _{i, s_T}([\psi])=r$ then $P_{i,\geq r} \psi \in T$ by Lemma \ref{maxconsprop}(5).

	- Suppose $\varphi = K_i  \psi$. Let $K_i  \psi \in T$. Since $\psi \in T/K_i$, then $\psi \in U$ for every $U$ such that $s_T \mathcal K_i  s_U$ (by the definition of $\mathcal K_i$). Therefore $(M^{*}, s_U)\models \psi$ by induction hypothesis ($\psi$ is subformula of $K_i \psi$), and then $(M^{*}, s_T)\models K_i  \psi$.
	
	Let $(M^{*}, s_T)\models K_i  \psi $. Assume  the opposite, that $K_i  \psi \not \in T$. Then $T/K_i \cup \{\neg \psi\}$ must be consistent. If it wouldn't be consistent, then  $T/K_i  \vdash \psi$ by \hyperref[deduc]{Deduction theorem} and $T \supset K_i(T/K_i) \vdash K_i \psi $ by Theorem \ref{necess}, ie. $ K_i  \psi \in T$, which is a contradiction. Therefore $ T/K_i  \cup \{\neg \psi\}$ can be extended to a maximal consistent $U$, so $s_T \mathcal K_i  s_U$. Since $\neg \psi \in U$,  then $(M^{*}, s_U)\models \neg \psi$ by induction hypothesis, so we get the contradiction $(M^{*}, s_T) \not \models_{M^{*}} K_i  \psi$.

	- Observe that  $\varphi =E_G \psi \in T $ iff $K_i \psi \in T$ $\, \mbox{ for all }  i \in G$ (by  Lemma \ref{max}(1)) iff $(M^{*}, s_T)\models K_i  \psi $ $ \mbox{ for all }  i \in G$ (by previous case) ie. $(M^{*}, s_T)\models E_G \psi $ (by the definition of $\models$ relation). 
	
	- $\varphi =C_G \psi \in T $ iff  $(E_G)^m \psi \in T$ $\, \mbox{ for all }  m \in \mathbb N$ (by  Lemma \ref{max}(3)) iff $(M^{*}, s_T)\models(E_G)^m \psi $ $\, \mbox{ for all }  m \in \mathbb N$ (by previous case) ie. $(M^{*}, s_T)\models C_G \psi $.
	
	- $\varphi =E_G^r \psi \in T $ iff $K_i^r  \psi = K_i(P_{i,\geq r}{\psi}) \in T$ $\, \mbox{ for all }  i \in G$ (by  Lemma \ref{max}(2)) iff $(M^{*}, s_T)\models K_i(P_{i,\geq r}{\psi})  $ (by the previous case  $\varphi = K_i  \psi$), i.e., $(M^{*}, s_T)\models E_G^r \psi$.
	
	- Let $\varphi =(F_G^r)^m \psi $. Since $(F_G^r)^0 \psi = \top$, the claim holds trivially. Also $\varphi =(F_G^r)^{m+1} \psi = E_G^r(\psi \wedge (F_G^r)^{m} \psi) \in T$ iff  $(M^{*}, s_T)\models E_G^r(\psi \wedge (F_G^r)^{m} \psi)$ (by the previous case) ie. $(M^{*}, s_T)\models (F_G^r)^{m+1} \psi  $, $m \in \mathbb N$.
	
	- $\varphi =C_G^r \psi \in T $ iff  $(F_G^r)^m \psi \in T$ $\, \mbox{ for all }  m \in \mathbb N$ (by  Lemma \ref{max}(4)) iff $(M^{*}, s_T)\models(F_G^r)^m \psi $ $\, \mbox{ for all }  m \in \mathbb N$ (by the previous case), i.e., $(M^{*}, s_T)\models C_G \psi $.

\end{proof}

From Lemma \ref{lemaPres} and Lemma \ref{Truth} we immediately obtain the following corollary.

\begin{theorem}\label{measurable M*}
	$M^*\in \mathcal{M}^{MEAS}_\mathcal{A}$.
\end{theorem}

\subsection{Completeness theorem}
Now we state the main result of this paper.
In the following theorem, we summarize the results obtained above in order to prove the strong completeness of our axiomatic system for the class of measurable models.

\begin{theorem}[Strong completeness theorem]\label{thm compl} A theory $T$ is consistent if and only if it is satisfiable in an $\mathcal{M}^{MEAS}_\mathcal{A}-$model. \end{theorem}

\begin{proof}  
	The direction from right to left is a consequence of Soundness theorem. 
	For the other direction, suppose that $T$ is a consistent theory. We will show that $T$ is satisfiable in the canonical model $M^*$, which belongs to $\mathcal{M}^{MEAS}_\mathcal{A}$, by Theorem \ref{measurable M*}. By Theorem \ref{lindenbaum}, $T$ can be extended to a saturated theory $T^*$. 
	From Lemma \ref{Truth} we have that $\varphi \in V   \mbox{ iff }  (M^{*}, s_{V}) \models \varphi$, for every saturated theory $V$. Consequently, 
	$(M^*, s_{T^*}) \models \varphi$, for every $\varphi \in T^*$, and therefore    $(M^*, s_{T^*}) \models T$.
\end{proof}

\section{Adding the consistency condition}\label{sec con}

In the logic $PCK^{fo}$ presented in this paper, we proposed the most general case, where no relationship is posed between the modalities for knowledge and probability. Indeed, in the definition of the probability spaces $\mathcal P (i,s) = (S_{i, s}, \chi_{i, s}, \mu_{i, s})$ the sample space of possible events $S_{i,s}$ is an arbitrary nonempty subset of the set of all states $S$.

Now we consider a natural additional assumption, called \emph{consistency condition} in \cite{KP}, which forbids an agent to place a positive probability to the event she knows to be false. This assumption can be semantically captured by adding the condition $S_{i,s} \subseteq \mathcal K_i(s)$ to Definition \ref{def model}. In the following definition we introduce the corresponding subclass of measurable models $\mathcal M_{\mathcal A}^{MEAS,CON}$.

\begin{definition}
	$\mathcal M_{\mathcal A}^{MEAS,CON}$ is the class of all measurable models $M=(S,D,I,\mathcal  K, \mathcal  P)\in\mathcal M_{\mathcal A}^{MEAS}$, such that 
	$$
	S_{i,s} \subseteq \mathcal K_i(s)
	$$
	for all $i$ and $s$, where $\mathcal P (i,s) = (S_{i, s}, \chi_{i, s}, \mu_{i, s})$.
\end{definition}

\noindent We will prove that adding the axiom\footnote{This type of axiom is standard in logics in  which probability is seen as an approximation of other modalities; for example, in probabilistic temporal logic, the axiom $G\varphi \to P_{ \geq 1}{\varphi}$ (``if  $\varphi$ always holds, then its probability is equal to 1'') is a part of axiomatization \cite{DBLP:journals/logcom/Ognjanovic06}.}\\

\hspace{4ex}CON. $K_i \varphi \to P_{i, \geq 1}{\varphi}$\\

\noindent to our axiomatization results with a system which is complete for the class of models  $\mathcal M_{\mathcal A}^{MEAS,CON}$. Note that in that case we can remove Probabilistic Necessitation from the list of inference rules since, in presence of CON, it is derivable from Knowledge Necessitation. Indeed, the applications of the rules RK and RP are restricted to theorems only, so if $\vdash \varphi$, then $\vdash K_i \varphi$ by RK, and $\vdash P_{i, \geq 1}{\varphi}$ by CON. Thus,
$ \dfrac{ \varphi}{ P_{i, \geq 1}{\varphi} }$  
is derivable rule in the axiomatic system that  we propose in  the following definition.

\begin{definition}
	The axiomatization $Ax_{PCK^{fo}}^{CON}$  consists of all the axiom schemata and inference rules from $Ax_{PCK^{fo}}$ except RP and, in addition, it contains the axiom CON.
\end{definition}

The proposed axiomatic system is complete for the class of models $\mathcal M_{\mathcal A}^{MEAS,CON}$.

\begin{theorem}\label{thm con}
	The axiomatization  $Ax_{PCK^{fo}}^{CON}$ is strongly complete for the class of models $\mathcal M_{\mathcal A}^{MEAS,CON}$.
\end{theorem}
\begin{proof}
	The proof follows the idea of the proof of completeness of $Ax_{PCK^{fo}}$ for the class of models $\mathcal M_{\mathcal A}^{MEAS}$ presented above. Similarly as it is done in Section \ref{sect lindenbaum}, we can show that any consistent theory $T$ can be extended to a saturated theory in $Ax_{PCK^{fo}}^{CON}$ (not that the saturated theories in $Ax_{PCK^{fo}}^{CON}$ and $Ax_{PCK^{fo}}$ don't coincide; for example, the formula $K_i\varphi\wedge P_{i,<1}\varphi$ is consistent for the former axiomatization, but it is inconsistent for the later one). Then we can construct the canonical model $M^{*} = (S,D,I, \mathcal{K}, \mathcal{P})$ using the saturated theories, and prove Truth lemma as in  Section  \ref{sec canonical}, and prove that $(M^*, s_{T^*}) \models T$ in the same way as in the proof of Theorem \ref{thm compl}. 
	
	The problem is that $M^*$ doesn't belong to the class  $\mathcal M_{\mathcal A}^{MEAS,CON}$, since the condition $S_{i,s} \subseteq \mathcal K_i(s)$ is not ensured. Nevertheless, we can use $M^*$ to obtain a model $M^{*'}$ from $M_{\mathcal A}^{MEAS,CON}$, in which $T$ is also satisfied. We define $M^{*'}$ by modifying only the probability spaces $\mathcal P (i,s) = (S_{i, s}, \chi_{i, s}, \mu_{i, s})$ from $M^*$ (i.e., $S,D,I $ and $\mathcal{K}$ are the same in both structures), in the following way:
	
	\smallskip
	
	$M^{*'}= (S,D,I, \mathcal{K}, \mathcal{P}')$, such that
	\begin{itemize}
		\item 
		$\mathcal P' (i,s) = (S'_{i, s}, \chi'_{i, s}, \mu'_{i, s})$, where 
		\begin{itemize}
			\item 
			$S'_{i, s} = S \cap \mathcal K_i(s)$
			\item  $\chi'_{i, s} = \{[\varphi]_{i, s} ' \, | \, \varphi \in Sent_{PCK}\}$, where $[\varphi]' = [\varphi]_{i, s}  \cap \mathcal K_i(s)$
			\item  if $ [\varphi]_{i, s}' \in \chi'_{i, s}$ then $\mu_{i, s}' ([\varphi]_{i, s}') =\mu_{i, s} ([\varphi]_{i, s})=\sup \, \{r  \, | \, P_{i,\geq r}{\varphi} \in s\}$.
			
		\end{itemize}
	\end{itemize}
	Now it only remains to prove that $M^{*'}$ is a model, i.e., that each $\mathcal P' (i,s)$ is a probability space,
	since the rest of proof is trivial: $S_{i,s} \subseteq \mathcal K_i(s)$ obviously holds, and $(M^{*'}, s_{T^*}) \models T$ is ensured by the construction of $\mathcal P'$, and the fact that $(M^*, s_{T^*}) \models T$.
	
	First we  show that every $\chi'_{i, s}$ is an algebra of sets, using the corresponding results from the proof of Lemma \ref{lemaPres}(2)
	\begin{itemize}
		\item $[\varphi]_{i, s}'\cup[\psi]_{i, s}' = ([\varphi]_{i, s}\cap \mathcal K_i(s) )\cup([\psi]_{i, s}\cap \mathcal K_i(s)) = ([\varphi]_{i, s}\cup[\psi]_{i, s})\cap \mathcal K_i(s) = [\varphi\vee\psi]_{i, s}\cap \mathcal K_i(s) = [\varphi\vee\psi]_{i, s}'\in \chi'_{i, s}$
		
		\item $S'_{i, s}\setminus[\varphi]_{i, s}' = S'_{i, s}\setminus([\varphi]_{i, s} \cap \mathcal K_i(s))$, so from  $S'_{i, s}= \mathcal K_i(s)$ and $[\varphi]_{i, s}=S\setminus [\neg\varphi]_{i, s}$ we obtain $S'_{i, s}\setminus[\varphi]_{i, s}' = [\neg\varphi]_{i, s} \cap \mathcal K_i(s) = [\neg\varphi]_{i, s}'\in \chi'_{i, s}$
	\end{itemize}
	Finally, we prove that  $\mu_{i, s}'$ is a finitely additive probability measure, for every $i$ and $s$.
	\begin{itemize}
		\item $\mu_{i, s}'(S'_{i, s})= ([\top]_{i, s}') =\mu_{i, s} ([\top]_{i, s})=1$
		
		\item In order to prove finite additivity of $\mu_{i, s}'$, we need to prove that 
		\begin{equation}\label{con eq 1}
		\mu_{i, s}' ([\varphi\vee \psi]_{i, s}')=\mu_{i, s}' ([\varphi]_{i, s}')+\mu_{i, s}' ([\psi]_{i, s}')
		\end{equation}
		whenever 
		\begin{equation}\label{con eq 2}
		[\varphi]_{i, s}'\cap[ \psi]_{i, s}'=\emptyset.
		\end{equation}
		The possible problem is that (\ref{con eq 2}) does not necessarily imply  $[\varphi]_{i, s}\cap[ \psi]_{i, s}=\emptyset$, so we cannot directly use finite additivity of $\mu_{i, s}$. On the other hand, we know that $\mu_{i, s} ([\varphi\vee \psi]_{i, s})=\mu_{i, s} ([\varphi]_{i, s})+\mu_{i, s} ([\psi]_{i, s})-\mu_{i, s} ([\varphi\wedge \psi]_{i, s})$. Since $\mu_{i, s}' ([\phi]_{i, s}')=\mu_{i, s} ([\phi]_{i, s})$ for every $\phi$, in order to prove (\ref{con eq 1}) it is sufficient to show that 
		\begin{equation}\label{con eq 3}
		([\varphi\wedge \psi]_{i, s})=0.
		\end{equation}
		From (\ref{con eq 2})  we obtain $ [\varphi]_{i, s}\cap[ \psi]_{i, s} \cap \mathcal K_i(s) =\emptyset$, i.e., $ [\varphi\wedge \psi]_{i, s} \cap \mathcal K_i(s) =\emptyset$. Consequently,  $ [\neg(\varphi\wedge \psi)]_{i, s} \subseteq\mathcal K_i(s)$, so $(M^*,t)\models \neg(\varphi\wedge \psi)$ for every $t\in \mathcal K_i(s)$, and  $(M^*,s)\models K_i\neg(\varphi\wedge \psi)$. Since $s$ is a saturated theory, from Truth lemma we have $ K_i\neg(\varphi\wedge \psi)\in s$, and consequently $ P_{i,\geq 1}\neg(\varphi\wedge \psi)\in s$, by CON. Then $\mu_{i, s} ([\varphi]_{i, s})=\sup \, \{r  \, | \, P_{i,\geq r}{\varphi} \in s\}=1$, which implies (\ref{con eq 3}).
	\end{itemize}
\end{proof}

\paragraph{Remark.}

Apart from consistency condition, Fagin and Halpern, \cite{KP} consider other relations between the sample space $S_{i,s}$ and possible worlds  $K_i(s)$, which model some typical situations in the multi-agent systems.  They also provide their characterization by the corresponding axioms. 

First they analyze the situations in which the probabilities of the events are common knowledge, i.e, there is a unique, collective and objective view on the probability of the events. Then the agents in the same state share the same known probability spaces, which is captured by the condition of \emph{objectivity}:
$\mathcal P (i,s) =\mathcal P (j,s)$ for all $i, j$ and $s$. 

Second, they model the situation where an agent uses the same probability space in all the worlds he considers possible. This situation occurs when no nonprobabilistic choices are made to cause different probability distributions in the possible worlds. The corresponding condition, called \emph{state determined property}, says that if $t \in \mathcal K_i(s)$, then $\mathcal P (i,s) =\mathcal P (i,t)$.

Third, sometimes the nonprobabilistic choices happen and induce varied probability spaces. Then the possible worlds  could be divided to partitions which share the same probability distributions, after such choice has been made. This case is specified by the condition of \emph{uniformity}:
if $\mathcal P (i,s) = (S_{i, s}, \chi_{i, s}, \mu_{i, s})$ and  $t \in \mathcal S_{i, s}$, then  $\mathcal P (i,s) =\mathcal P (i,t)$.

Similarly as we have done with consistency condition, we  can also characterize the three above mentioned conditions by adding corresponding axioms to our axiomatic system. It is straightforward to check that the following axioms, which are similar to the ones proposed in \cite{KP}, capture the mentioned relations between modalities of knowledge and probability: \\

\hspace{4ex}$P_{i,\geq r}{\varphi} \to P_{j,\geq r}{\varphi}$ (objectivity),

\hspace{4ex}$P_{i,\geq r}\varphi \to K_i P_{i,\geq r}\varphi$ (state determined property),

\hspace{4ex}$P_{i,\geq r}\varphi \to P_{i,\geq 1}{P_{i,\geq r}}\varphi$ (uniformity).\\

\section{Conclusion}

The starting points for our research were the  papers \cite{KP,InfiHalp} where weakly complete axiomatizations for a propositional logic combining knowledge and probability, and a non-probabilistic propositional logic for knowledge with infinitely many agents (respectively), are presented. We combine those two approaches and extend both of them to the logic $PCK^{fo}$ with an expressive first-order language.

We provide a sound and strongly complete axiomatization $Ax_{PCK^{fo}}$ for the corresponding semantics of $PCK^{fo}$.  Since 
any reasonable, semantically defined  first-order epistemic logic with common knowledge is not recursively axiomatizable \cite{Wolter}, we propose the axiomatization with  infinitary rules of inference, and we obtain completeness modifying the standard Henkin construction of saturated extensions of consistent theories.
In the logic $PCK^{fo}$ we consider the most general semantics, with independent modalities for knowledge and probability.  We also show  how to extend the set of axioms and modify the axiomatization technique in order to capture models in which agents assign probabilities only to the sets of worlds they consider possible.  We also give hints how to extend our axiomatization  in several different ways, to capture other interesting relationships between the  modalities for knowledge and probability, considered in \cite{KP}.

In this paper, we use the semantic definition of the probabilistic common knowledge operator $C_G^r$  proposed by Fagin and Halpern \cite{KP}. As we have mentioned in Section \ref{sec semantics}, Monderer and Samet \cite{Monderer} proposed a different definition, where probabilistic common knowledge is equivalent to the infinite conjunction of the formulas $E_G^r \varphi, (E_G^r)^2 \varphi, (E_G^r)^3 \varphi \dots $
It is easy to check that our axiomatization $Ax_{PCK^{fo}}$ can be easily modified in order to capture the definition of Monderer and Samet. Namely, the axiom APC and rule RPC should be replaced with the axiom 
%
$C_G^r\varphi \to (E_G^r)^m\varphi,\, m \in \mathbb N $ 
and the inference rule 
%
$  \dfrac{ \{ \Phi_{k, \boldsymbol{\uptheta}, \mathbf{X}}(E_G^r)^m \varphi)\, | \,  m \in \mathbb N\} }{  \Phi_{k, \boldsymbol{\uptheta}, \mathbf{X}}(C_G^r \varphi)}.$

\section*{Acknowledgment}

This work was supported  by the Serbian Ministry of Education and Science through projects ON174026, ON174010  and III44006, and by ANR-11-LABX-0040-CIMI.

\end{document}